\PassOptionsToPackage{dvipsnames}{xcolor}
\documentclass[journal,twoside,web]{ieeecolor}
\usepackage{generic}

\def\BibTeX{{\rm B\kern-.05em{\sc i\kern-.025em b}\kern-.08em
	T\kern-.1667em\lower.7ex\hbox{E}\kern-.125emX}}
\usepackage{graphicx}
\usepackage{tikz}
\usepackage{pgfmath}
\usepackage{pgfplots}
\usepgfplotslibrary{statistics}
\pgfplotsset{compat=1.18}

\pgfplotsset{%
	compat=1.18,
	my boxplot/.style={
		boxplot,
		mark=x,
		boxplot/every box/.style={
			solid,
			draw=NavyBlue,
			fill=Black!15!White,
			line width=0.3mm,
		},
		boxplot/every whisker/.style={
			solid,
			draw=Black,
			line width=0.3mm,
		},
		boxplot/every median/.style={
			solid,
			draw=BrickRed,
			line width=0.3mm,
		},
		every mark/.style={
			draw=BrickRed!60!white,
			mark size=2pt,
			line width=0.1mm,
		},
	}
}

\let\labelindent\relax
\usepackage[inline]{enumitem}
\usepackage[normalem]{ulem}

\usepackage{mathtools}
\usepackage{amssymb}
\usepackage{amsthm}

\theoremstyle{definition}
\newtheorem{remark}{Remark}

\theoremstyle{plain}
\newtheorem{theorem}{Theorem}
\newtheorem{proposition}{Proposition}
\newtheorem{corollary}{Corollary}
\newtheorem{lemma}{Lemma}

\newtheorem{definition}{Definition}

\newcommand{\R}{\mathbb{R}}
\newcommand{\N}{\mathbb{N}}

\newcommand{\st}{\,:\,}

\DeclareMathOperator{\diag}{diag}
\DeclareMathOperator{\col}{col}

\DeclarePairedDelimiter\norm{\vert}{\vert}
\DeclarePairedDelimiterX\inner[2]{\langle}{\rangle}{#1, #2}

\usepackage[hidelinks]{hyperref}

\usepackage[noadjust]{cite}

\begin{document}

\title{Kernel-Based Learning of \\ Stable Nonlinear Systems}

\author{
	Matteo Scandella,
	Michelangelo Bin, \IEEEmembership{Member, IEEE} and
	Thomas Parisini, \IEEEmembership{Fellow, IEEE}
	\thanks{This work has been supported by European Union's Horizon 2020 research and innovation programme under grant agreement no. 739551 (KIOS CoE).}
	\thanks{%
		Matteo Scandella is with the Department of Management, Information and Production Engineering, University of Bergamo, via Marconi 5, 24044, Dalmine (BG), Italy (email: \texttt{matteo.scandella@unibg.it}).
		Michelangelo Bin is with Department of Electrical, Electronic and Information Engineering, University of Bologna, Bologna, Italy (email: \texttt{michelangelo.bin@unibo.it}).
		During the conceptualization and first drafting of this article, M. Bin and M. Scandella were with the Department of Electrical and Electronic Engineering, Imperial College London, London SW7 2AZ, UK.
		Thomas Parisini is with the Department of Electrical and Electronic Engineering, Imperial College London, SW72AZ London, U.K.,
		He is also with the Department of Electronic Systems, Aalborg University, Denmark, and with
		the Department of Engineering and Architecture, University of Trieste, Italy
		(e-mail: \texttt{t.parisini@imperial.ac.uk}).
	}
}

\maketitle

\begin{abstract}
	Learning models of dynamical systems characterized by specific stability properties is of crucial importance in applications.
	Existing results mainly focus on linear systems or some limited classes of nonlinear systems and stability notions, and the general problem is still open.
	This article proposes a kernel-based nonlinear identification procedure to directly and systematically learn stable nonlinear discrete-time systems.
	In particular, the proposed method can be used to enforce, on the learned model, bounded-input-bounded-state stability, asymptotic gain, and input-to-state stability properties, as well as their incremental counterparts.
	To this aim, we build on the reproducing kernel theory and the Representer Theorem, which are suitably enhanced to handle stability constraints in the kernel properties and in the hyperparameters' selection algorithm.
	Once the methodology is detailed, and sufficient conditions for stability are singled out, the article reviews some widely used kernels and their applicability within the proposed framework.
	Finally, numerical results validate the theoretical findings showing, in particular, that stability may have a beneficial impact in long-term simulation with minimal impact on prediction.
\end{abstract}


\section{Introduction}
\label{sec:introduction}

\IEEEPARstart{A}{} key issue in system identification is to learn models that not only fit the observations, but also possess specific stability properties~\cite{Ljung2010a, Sjoeberg1995a}.
Indeed, stability is desirable in many applications as it provides robustness guarantees when the model is used for prediction or simulation, especially with a long time horizon.
However, most existing learning procedures do not guarantee that the learned model is stable, even when the observations are generated by a stable system.
In this article, we propose a learning approach for nonlinear systems that, instead, systematically guarantees a desired stability property on the learned models.

While classical identification and machine learning approaches do not deal with the problem of imposing stability on the learned models~\cite{Box2015a,Pintelon2012a, vanOverschee2012a}, in the case of linear time-invariant (LTI) systems such a problem is well-studied and many solutions exist.
Among others, it is worth mentioning methods based on ARMAX (AutoRegressive Moving Average eXogenous input) models~\cite{Soderstrom1981a,Nallasivam2011a}, subspace~\cite{Lacy2003a, Umenberger2016a, Umenberger2018a} and set-membership~\cite{Cerone2011a, Lauricella2020a} methods, estimation techniques working in the frequency domain~\cite{Lataire2016a}, and kernel-based approaches~\cite{Pillonetto2014a, Scandella2022a}.
Instead, in a nonlinear framework the problem is considerably less explored and fewer results are available;
moreover, the learning problem is exacerbated by the many available notions of stability and by the wide range of modeling techniques, which have led to more scattered results mainly tailored to particular cases.
For instance, there are specific methods guaranteeing stability for linear parameter-varying systems~\cite{Darwish2018a, Cerone2012a}, nonlinear finite impulse response systems~\cite{Qin1996a, Pillonetto2018a}, and linear switching systems~\cite{DeIuliis2022a}.
For more generic nonlinear structures, some works rely on neural networks~\cite{Rubio2007a, Bonassi2021a, Bonassi2022a};
more in detail~\cite{Bonassi2021a, Bonassi2022a} constrain the parameters of a neural network to impose different stability properties, although the stability constraints can easily be violated due to numerical problems arising from the imposition of constraints during the training phase.
Other approaches to guarantee stability in the nonlinear case rely on Koopman operators~\cite{Bevanda2022a, Khosravi2023c}, linear matrix inequality constraints~\cite{Umenberger2019a,Shakib2023a}, or
Gaussian process state-space models~\cite{Umlauft2020a,Xiao2020a,Beckers2016a};
in particular,~\cite{Umlauft2020a,Xiao2020a} develop a two-steps procedure that first identifies a potentially unstable model fitting the observations, and then learns a virtual control law modifying the previously-estimated model to guarantee stability.
Finally,~\cite{vanWaarde2022a, vanWaarde2023a} propose a kernel-based procedure to learn contractive maps between function spaces, thereby learning stable input-output models.

In this article, we develop a kernel-based  learning approach to directly and systematically learn nonlinear stable discrete-time models.
The proposed approach is particularly inspired by~\cite{vanWaarde2022a, vanWaarde2023a}, although here we focus on predictors and state-space models instead of input-output maps, and we target more general stability notions.
Moreover, unlike most of the aforementioned works, the proposed method is systematic, is not tailored on a specific model structure, and the stability constraints are robust to numerical errors in the solution procedure.
Furthermore, we target multiple notions of stability of primary interest in control theory, i.e., bounded-input-bounded-state (BIBS) stability~\cite{Andriano1997a}, asymptotic gain (AG)~\cite{Sontag1996a}, input-to-state stability (ISS)~\cite{sontag_smooth_1989,jiang_input--state_2001}, and their incremental counterparts $\delta$BIBS stability, $\delta$AG, and $\delta$ISS~\cite{Angeli2002a}.
The enforced stability guarantees refer both to the exogenous input of the system and to the prediction error.
The former is the main motivation of this article while, as commented in Section~\ref{sec:problem} and shown in the numerical validation, stability with respect to the prediction error may lead to a considerable benefit when the learned model is used for simulation over a long (potentially infinite) time horizon.

Kernel-based approaches are popular in system identification and machine learning because they allow to systematically enforce desired properties on the learned models by appropriately shaping the kernel function~\cite{Dinuzzo2015a,Micchelli2006a,Pillonetto2011a} or the regularization term~\cite{Mazzoleni2019b,Formentin2019a,Formentin2021a}.
More specifically, for LTI models, kernel-based techniques enable identifying bounded-input-bounded-output (BIBO) stable models without a priori knowledge of the model's order~\cite{Pillonetto2010a,Pillonetto2014a,Scandella2022a}.
The drawback of these methods is that they exploit the properties of the impulse response of an LTI system and its relation to BIBO stability, thereby ruling out possible extensions to the nonlinear case.
Existing kernel-based approaches for nonlinear systems are instead based on the estimation of input-output operators~\cite{vanWaarde2022a, vanWaarde2023a, DallaLibera2021a} or, alternatively, on the direct learning of the predictor function of the state~\cite{Umlauft2020a,Xiao2020a} or of the output~\cite{Pillonetto2011a, Pillonetto2018a, Mazzoleni2020b, Bhujwalla2016a, DallaLibera2021a} instead of the full model.
The methodology proposed in this article falls into the latter category.
In particular, we equip the technique used in~\cite{Pillonetto2011a,Mazzoleni2020b} with constraints on the learning algorithm and the hyperparameters' selection method.
In this way, the desired stability properties are guaranteed on the learned model while, at the same time, the predictor selection is optimized by fitting the available dataset.
Moreover, since the proposed methodology is an extension of the classical kernel approach, it is not tied to a single specific kernel structure nor to a specific hyperparameters' selection method. Instead, it can easily be adapted to different model-assessment rationales, e.g., empirical Bayes~\cite{MacKay1992a} or cross-validation techniques~\cite{Golub1979a,Allen1974a,Stone1974a}, and can be used with different kernel functions.

However, the price to pay for such a generality is that the procedure may lead to a complex bilevel optimization problem that, in some cases, can be challenging to solve numerically.
Moreover, the conditions enforced to guarantee stability may be conservative in general, and their enforcement may lead to a loss of prediction performance compared to the unconstrained (but possibly unstable) case.
These drawbacks worsen with the predictor's order, so as the proposed method may not scale well with the models' dimension.
All these critical aspects, however, concern individual parts of the overall procedure that can be addressed in specific application scenarios without affecting the general methodology.

\subsubsection*{Paper organization}
Section~\ref{sec:problem} lays down the problem formulation.
Section~\ref{sec:method} briefly recalls the regularized kernel method, and introduces the proposed methodology as an extension of it.
Section~\ref{sec:stability} shows how the proposed method can be used to learn BIBS stable and ISS models.
Similarly, Section~\ref{sec:δstability} considers the case of $\delta$BIBS stability and $\delta$ISS.
Section~\ref{sec:kernels} analyzes some popular kernels within the proposed framework.
Finally, Section~\ref{sec:sim} presents a numerical validation of the proposed approach.

\subsubsection*{Notations}
$\R$ and $\N$ denote the set of real and natural numbers, respectively ($0 \in \N$).
Given $n,m \in \N$, $\R^{n \times m}$ denotes the set of $n\times m$ matrices and $\N_{\ge m} \coloneq \{ x \in \N \st x \ge m \}$.
Tuples of real numbers and column vectors are used interchangeably.
Given $n \in \N$, $0_n \coloneq (0,\dots, 0) \in \R^n$, and $I_n \in \R^{n \times n}$ denotes the identity matrix.
Given two matrices $A$ and $B$, $\col(A,B)$ and $\diag(A,B)$ denote, respectively, their vertical and diagonal concatenations whenever they are well-defined.
$\norm*{x}$ denotes the Euclidean norm of $x \in \R^n$.
Let $T\subseteq\N$; we denote by $\R^T$ the set of real-valued sequences $x:T \to \R$ indexed by $T$.
For each $x \in \R^T$ and $i,j \in T$, $x_i$ is the $i$th element of $x$ and $x_{j:i} \coloneq (x_j, \ldots , x_i)$.
Furthermore, $\norm*{x}_\infty \coloneq \sup_{t\in T} \norm*{x_t} \in [0, \infty]$.
A function $\gamma : [0, \infty) \to [0, \infty)$ is of class $\mathcal{K}$ if it is continuous, strictly increasing, and $\gamma(0)=0$.
A function $\beta : [0, \infty) \times [0, \infty) \to [0, \infty)$ is of class $\mathcal{KL}$ if, for each $t \in [0, \infty)$, $\beta(\cdot, t)$ is of class $\mathcal{K}$, and for each $s\in[0,\infty)$, $\beta(s,\cdot)$ is continuous,  decreasing, and $\lim_{t\to\infty} \beta(s, t) = 0$.

\section{Problem Formulation}
\label{sec:problem}

Given a discrete-time system $\Psi \subseteq \R^{\N} \times \R^{\N}$ relating \emph{input} sequences $u \in \R^\N$ to \emph{output} sequences $y \in \R^\N$,
we consider the general problem of estimating from data a \emph{stable} model of $\Psi$ that can be used for simulation and control purposes.
While several \emph{simulation error methods (SEM)} can be found in the literature that optimize the model selection for simulation~\cite{Pillonetto2025a,Forgione2021a,Piroddi2008a,Krikelis2024a}, these techniques come with considerable issues:
\begin{enumerate*}[label={\roman*)}]
	\item the additional complexity of the underlying optimization problem (which is not convex),
	\item the need for a fixed simulation horizon,
	\item the consequent lack of guarantees of error accumulation in infinite-horizon simulation, and
	\item the absence of stability guarantees.
\end{enumerate*}
Instead, we approach the problem at hand as a \emph{prediction-error estimation method with stability constraints}.
In this way, we retain a simpler selection method, and we are able to impose stability constraints that also provide guarantees on error accumulation during simulation over infinite time horizons.

More specifically, given an input-output pair $(u,y) \in \Psi$, we define the \emph{prediction sequence} $\hat{y}^{\mathrm{pre}} \in \R^\N$ such that,
\begin{equation*}
	\forall t \in \N_{\ge m},
	\hspace{3ex}
	\hat{y}^{\mathrm{pre}}_t \coloneq f( \xi_t, u_t ),
\end{equation*}
where $\hat{y}^{\mathrm{pre}}_t$ is the prediction of $y_t$, $f : \R^{2m+1} \to \R$ is a function to be selected based on the available data, $m \in \N$, and, for each $t \in \N_{\ge m}$, $\xi_t \coloneq ( y_{t-m : t-1}, u_{t-m : t-1} ) \in \R^{2m}$ collects past values of the output $y$ and the input $u$.
We refer to $m$ as the \emph{model order}, to $f$ as the \emph{predictor}, and we define the \emph{prediction error} sequence as $e \coloneq y - \hat{y}^{\mathrm{pre}} \in \R^\N$.
Once a predictor $f$ is selected, an estimated model $\hat{\Psi}$ of $\Psi$ employable for the purpose of simulation and control can be simply obtained using the predictor $f$ iteratively as a multiple-steps ahead prediction dynamics.
Namely, given any input $u$ of System $\Psi$,  a sequence $\hat{y} \in \R^\N$ is produced by the estimated model $\hat{\Psi}$ if, for each $t \in \N_{\ge m}$, $\hat{y}_t \coloneq f ( \hat{\xi}_{t}, u_{t} )$ where $\hat{\xi}_{t} \coloneq ( \hat{y}_{t-m: t-1}, u_{t-m: t-1} )$.

The aim of the paper is to devise a method to select $f$ in such a way that $\hat{\Psi}$ is a good model and, in addition, as a system with input $u$ and output $\hat{y}$, is stable in some suitable sense.
The specific notions of stability on which this paper focuses are formally defined later in Sections~\ref{sec:stability} and~\ref{sec:δstability} by making reference to an \emph{auxiliary dynamical system} defined as follows
\begin{equation*}
	\Sigma(f) \, :
	\hspace{3ex}
	x_{t+1} = A x_t + G f( x_t, v_t ) + G \varepsilon_t + H v_t,
\end{equation*}
in which $x_t\in\R^{2m}$ is the state, $\varepsilon_t \in \R$ and $v_t \in \R$ are two exogenous inputs, and where $A \coloneq \diag(S,S)$, $G \coloneq \col(B,0_m)$, and $H \coloneq \col(0_m, B)$, with
\begin{align*}
	S & \coloneq
	\begin{bmatrix}
		0_{m-1} & I_{m-1}\\
		0 & 0^\top_{m-1}
	\end{bmatrix},
	&
	B & \coloneq
	\begin{bmatrix}
		0_{m-1} \\
		1
	\end{bmatrix}.
\end{align*}
In formalizing the required stability notions, we use $\Sigma(f)$ instead of directly $\hat{\Psi}$ to capture, at the same time, both the stability requirements we want to enforce on the estimated model $\hat{\Psi}$ with respect to the input $u$, and also the stability properties that lead to guarantees on the lack of accumulation of prediction errors in simulation.
Both points are illustrated hereafter.

First, we notice that, by construction of the matrices $A$, $G$ and $H$, the sequence $\hat{\xi}$ in a solution of $\Sigma(f)$ with $\varepsilon=0$ and $v=u$.
Namely, $\Sigma(f)$ with $\varepsilon=0$ and $v=u$, and with output $\hat{y}_t = f(x_t,u_t)$, is a \emph{state-space model of $\hat{\Psi}$}.
As a consequence, requiring stability of $\Sigma(f)$ with respect to the initial conditions and to the input $v$ implies an analogous stability property for the estimated model $\hat{\Psi}$ with respect to $u$.

Next, regarding stability guarantees for the lack of prediction-error accumulation in simulation, we notice that, for every $(u,y) \in \Psi$ and every $t \ge m$, the sequence $\xi$ is a solution of $\Sigma(f)$ with $v=u$ and $\varepsilon=e$ because, by the definition of $e$, $y_t = f(\xi_t, u_t) + e_t$ for every $t\in\N$.
Hence, $\Sigma(f)$, with inputs $v=u$ and $\varepsilon=e$, and output $y_t = f(x_t,u_t) + e_t$, is a \emph{state-space model of $\Psi$}.
Therefore, $\Sigma(f)$ with $v=u$ links the process $\Psi$ with its estimated model $\hat{\Psi}$ as, for different inputs $\varepsilon$, it describes both of them.
Since the input $v$ is the same in both cases, the difference between $y$ and its estimate $\hat{y}$ can then only be attributed to the different initialization and the different choice of $\varepsilon$.
As a consequence, requiring stability of $\Sigma(f)$ with respect to the input $\varepsilon$ means providing deterministic guarantees on the extent by which past prediction errors may affect future ones.
To better illustrate this latter point, consider the case in which $\Sigma(f)$ is $\delta$ISS (see Definition~\ref{def:δISS}).
Since $\xi$ is a solution of $\Sigma(f)$ with $v=u$ and $\varepsilon=e$, and $\hat{\xi}$ is a solution of $\Sigma(f)$ with $v=u$ and $\varepsilon=0$, then $\delta$ISS guarantees that $\xi - \hat{\xi}$, and thus $y - \hat{y}$, is bounded by a continuous function of the \emph{maximum} prediction error attained before $t$.
Hence, past prediction errors do not accumulate indefinitely during simulation, as only the maximum prediction error matters in the bound.

To learn a predictor $f$ that fits the available data and satisfies the stability constraints, we assume to have at our disposal a \emph{dataset} $\mathcal{D} \coloneq (\bar{u}_i, \bar{y}_i)_{i=1}^n$, which is a collection of $n \in \N_{> m}$ data points $(\bar{u}_i, \bar{y}_i)$ taken from an input-output pair of $\Psi$ starting from an unknown time $t_0 \in \N$.
In particular, for all $i=\{1, \ldots, n\}$, $(\bar{u}_i, \bar{y}_i) \coloneq (u_{t_0 + i - 1}, y_{t_0 + i - 1})$, in which $(u,y) \in \Psi$.
Then, in qualitative terms, our goal is to devise an algorithm that selects a predictor $f : \R^{2m+1} \to \R$ in such a way that
\begin{enumerate*}[label={\textbf{(P\arabic*)}}, ref=\textbf{(P\arabic*)}]
	\item \label{prob:main:small} for every input-output pair $(u,y) \in \Psi$, the corresponding error sequence $e$ is ``small'', and
	\item \label{prob:main:stability} $\Sigma(f)$ has a given desired stability property with respect to both $v$ and $\varepsilon$.
\end{enumerate*}
The fitting performance requirement~\ref{prob:main:small} will be approached in Section~\ref{sec:method} as a regularized least-squares problem, as is typical in kernel-based identification.
The stability requirement~\ref{prob:main:stability} is instead addressed by the insertion of suitable constraints in the identification procedure, and it will refer to several specific stability properties.
In particular, Section~\ref{sec:stability} focuses on BIBS stability and ISS, while Section~\ref{sec:δstability} is devoted to their incremental counterparts $\delta$BIBS stability and $\delta$ISS.

\section{Methodology}
\label{sec:method}

This section introduces the proposed learning method as an extension of the kernel-based regularized regression for the selection of the predictor $f$ enabling the possibility of enforcing desired stability conditions on $\Sigma(f)$.
Here, we only present the general methodology, deferring the embedding of specific stability properties to the subsequent sections.

\subsection{Regularized kernel-based learning}
\label{sec:method:classic}

The kernel-based regularized regression is a well-established technique in both machine learning~\cite{Scholkopf2018a, Vapnik1998a, Wahba1990a} and system identification~\cite{Pillonetto2011a, Pillonetto2014a, Lataire2017a, Pillonetto2022b, Chen2012a} selecting the predictor function $f$ within a hypothesis set of candidates according to a predetermined criterion, usually expressed as an optimization problem.
The hypothesis set is a Reproducing Kernel Hilbert Space (RKHS)~\cite{Aronszajn1950a, Saitoh2016a} generated by a function $k_\eta : \R^{2m+1} \times \R^{2m+1} \to \R$ defined by a parameter $\eta \in \Phi_k \subseteq \R^{n_\eta}$ ($n_\eta \in \N$).
We denote by $\mathcal{H}(k_\eta)$ the RKHS associated with $k_\eta$, while $\norm*{\,\cdot\,}_{k_\eta}$ and $\inner*{\cdot}{\cdot}_{k_\eta}$ denote the norm and the inner product on $\mathcal{H}(k_\eta)$, respectively.
We refer to $k$ as the \emph{kernel structure}, and to $k_\eta$ as the actual \emph{kernel function} defined using the parameter $\eta$ and the structure $k$.
Furthermore, we restrict $\Phi_k$ to contain only values of $\eta$ that make $k_\eta$ a valid kernel function, namely, such that $k_\eta$ is symmetric and positive semidefinite~\cite{Rasmussen2006a,Aronszajn1950a}.
Given $z \in \R^{2m+1}$, we let $k_\eta^z : \R^{2m+1} \to \R$ be the \emph{representer function} (or \emph{kernel slice}) of $z$, defined as $k_\eta^z (\cdot) \coloneq k_\eta(\cdot, z)$.

Given a kernel structure $k$ and a dataset $\mathcal{D}$, the predictor $f$ is chosen as the solution of the optimization problem
\begin{equation} \label{eq:opt_c:classic}
	\min_{f \in \mathcal{H}(k_\eta)}
	\hspace{3ex}
	\sum_{i=m+1}^n \big( f(\bar{z}_i) - \bar{y}_i \big)^2 + \beta \norm*{f}_{k_\eta}^2,
\end{equation}
where $(\beta, \eta) \in (0, \infty) \times \Phi_k$ are parameters to be tuned and, for each $t \in \{m+1, \ldots ,n\}$, $\bar{z}_t \coloneq ( \bar{y}_{t-m : t-1}, \bar{u}_{t-m : t} ) \in \R^{2m+1}$ is taken from the dataset $\mathcal{D}$.
The first term of the cost function weights how well the model fits the dataset.
Hence, selecting $f$ according to~\eqref{eq:opt_c:classic} addresses~\ref{prob:main:small} as is typical in kernel-based methods~\cite{Pillonetto2011a,Mazzoleni2022a}.
Instead, the second term is a regularization term balancing the complexity of $f$ and the adherence to the available data.
It is well-known that~\eqref{eq:opt_c:classic} has a unique solution~\cite[Prop.~8]{DeVito2004a} of the form
\begin{equation} \label{eq:def:f_c}
	f \coloneq \sum_{i=m+1}^n c_i k_\eta^{\bar{z}_i} \ \in \ \mathcal{H}(k_\eta) \, ,
\end{equation}
for some $c = c_{m+1:n} \in \R^{n-m}$.
In particular, the following result holds.
\begin{theorem}[Representer Theorem~\cite{Argyriou2014a}]
	\label{thm:representer:classic}
	There exists $c \in \R^{n-m}$ such that $f$, defined in~\eqref{eq:def:f_c}, solves~\eqref{eq:opt_c:classic}.
\end{theorem}
More specifically, by using the properties of the RKHS and Theorem~\ref{thm:representer:classic}, it is possible to show that the optimal solution is obtained with $c = \big(K_\eta + \beta I_{n-m}\big)^{-1} \bar{y}$, where $K_\eta \in \R^{(n-m) \times (n-m)}$ is the matrix with the $(i,j)$th element given by $k_\eta(\bar{z}_{m+i}, \bar{z}_{m+j})$, and $\bar{y} \coloneq \bar{y}_{m+1:n} \in \R^{n-m}$.

The shape of the kernel function $k_\eta$ determines the properties of the selected predictor $f$.
By suitably choosing the kernel structure $k$, one can embed in all the functions of the RKHS (hence, in $f$) some prior knowledge such as continuity, smoothness, or square-integrability.
In the case of linear systems, the problem of selecting a good kernel structure is well-studied, and it is possible to select a kernel structure guaranteeing desirable properties including asymptotic stability~\cite{Dinuzzo2015a,Chen2012a}.
The nonlinear case is instead less studied (see, e.g.,~\cite{Pillonetto2011a, Umlauft2020a, Mazzoleni2020b}), and the proposed kernels are typically not as well-motivated as in the linear case.
Moreover, it is typically not possible to embed stability constraints directly within $k$ (see Remark~\ref{rem:k_stability}).
Instead, as it will be clarified in the following, the key to impose stability turns out to be the regularization term in~\eqref{eq:opt_c:classic}, which favors functions with smaller norm.

Determining $f$ by solving~\eqref{eq:opt_c:classic}, requires the selection of the hyperparameters $\beta$ and $\eta$.
Their value is of primary importance because it determines the trade-off between complexity of the model and fitting performance, thereby causing or avoiding over- and under-fitting.
In the literature, this problem is also known as \emph{model selection} (or \emph{assessment}), and it is historically solved using techniques based on cross-validation~\cite{Herzberg1969a,Rasmussen2006a}.
In particular, \emph{k-fold cross-validation}~\cite{Stone1974a,Rasmussen2006a}, \emph{leave-one-out cross-validation}~\cite{Allen1974a,Wahba1990a,Rasmussen2006a}, and \emph{generalized cross-validation}~\cite{Golub1979a} are popular methods that can be used in a wide range of applications.
More recently, also \emph{Empirical Bayes (EB)}~\cite{MacKay1992a} methods have shown to be very efficient~\cite[Sec.~5.4.1]{Rasmussen2006a}.
All the previous techniques have successfully been used in kernel-based identification for both linear systems~\cite{Pillonetto2015b,Mu2018a,Mu2018b} and nonlinear ones~\cite{Pillonetto2011a, Mazzoleni2022a, Formentin2019a} when no guarantee is provided on the stability properties of the identified models, and they all boil down to the minimization of a cost function.
Indeed, in general, one can cast the model selection problem in terms of the following optimization problem
\begin{equation} \label{eq:hp_sel_gen:classic}
	\min_{\beta \in [\iota, \infty), \, \eta \in \Phi_k}
	\hspace{3ex}
	J(\beta, \eta),
\end{equation}
in which the cost function $J : (0,\infty) \times \Phi_k \to \R$ defines the actual method implemented, and $\iota \in (0, \infty)$ is a parameter introduced to make the feasibility set closed.
It is important to notice that $J$ in~\eqref{eq:hp_sel_gen:classic} typically depends on the dataset $\mathcal{D}$, and often explicitly on the solution of~\eqref{eq:opt_c:classic}.
Therefore,~\eqref{eq:opt_c:classic} and~\eqref{eq:hp_sel_gen:classic} often form a \emph{bilevel optimization problem}~\cite{bard_practical_1998} whose numerical solution may require a complex iterative procedure solving~\eqref{eq:opt_c:classic} and~\eqref{eq:hp_sel_gen:classic} multiple times.

\subsection{Modifying the kernel method for stability guarantees}
\label{sec:method:stable}

The method explained so far in Section~\ref{sec:method:classic} only concerns Point~\ref{prob:main:small}, but it does not address Point~\ref{prob:main:stability}, since no specific stability property is a priori guaranteed on the learned model $\Sigma(f)$.
To address Point~\ref{prob:main:stability}, we modify the optimization problems~\eqref{eq:opt_c:classic} and~\eqref{eq:hp_sel_gen:classic} to enforce a specific stability property on the learned model.
In particular, we first restrict the kernel hyperparameter $\eta$ to belong to a set $\Omega_k \subseteq \Phi_k$, called \emph{viability set}, whose definition determines the actual desired stability property on $\Sigma(f)$.
The construction of the viability set $\Omega_k$ is discussed in detail in Sections~\ref{sec:stability} and~\ref{sec:δstability}.
Once $\Omega_k$ is given, the hyperparameter selection problem~\eqref{eq:hp_sel_gen:classic} is modified as follows
\begin{equation} \label{eq:hp_sel_gen:stable}
	\min_{\beta \in [\iota, \infty), \, \eta \in \Omega_k}
	\hspace{3ex}
	J(\beta, \eta).
\end{equation}
It is worth noticing that the proposed approach does not rely on a specific cost function $J$ used in~\eqref{eq:hp_sel_gen:classic}.
Hence, it is possible to use all the various model assessment methods discussed in Section~\ref{sec:method:classic} that can be written in the form of~\eqref{eq:hp_sel_gen:classic}.

After solving~\eqref{eq:hp_sel_gen:stable}, we modify the optimization problem~\eqref{eq:opt_c:classic} by adding a constraint on the norm of $f$.
In particular, we consider the following constrained problem
\begin{equation} \label{eq:opt_c:stable}
	\left\{
		\begin{aligned}
			&\min_{f \in \mathcal{H}(k_\eta)}
			\hspace{3ex}
			\sum_{i=m+1}^n \big( f(\bar{z}_i) - \bar{y}_i \big)^2 + \beta \norm*{f}_{k_\eta}^2
			\\
			&\hspace{2ex}\textnormal{subject to}
			\hspace{3ex}
			m \norm*{f}_{k_\eta}^2 \le \chi
		\end{aligned}
	\right. \, ,
\end{equation}
in which $\chi \in (0,1)$ is a parameter less than, but close to $1$ needed to ensure that the feasibility set is closed (which is required for the numerical solution of the optimization problem~\cite[Ch.~2]{bertsekas_nonlinear_2016}).

Theorem~\ref{thm:representer:stable_c} below (proved in Appendix~\ref{thm:representer:stable_c::proof}) extends Theorem~\ref{thm:representer:classic} to the above constrained case by guaranteeing that~\eqref{eq:opt_c:stable} always has at least one solution and by providing the analytical expression of one of them.
This theorem uses the function $\gamma : (0, \infty) \to \R$ defined by
\begin{equation} \label{eq:def:gamma}
	\gamma(\alpha)
	\coloneq
	m \bar{y}^\top \left( K_\eta + \alpha I_n \right)^{-1} K_\eta \left( K_\eta + \alpha I_n \right)^{-1} \bar{y} - \chi.
\end{equation}
\begin{theorem} \label{thm:representer:stable_c} %
	Problem~\eqref{eq:opt_c:stable} is solvable, and a solution is given by~\eqref{eq:def:f_c} with
	\begin{equation} \label{eq:vec_c:stable}
		c = \big(K_\eta + \max(\bar{\alpha}, \beta) I_{n-m}\big)^{-1} \bar{y}
	\end{equation}
	in which $\bar{\alpha} \in [0, \infty)$ is such that $\gamma(\bar{\alpha}) = 0$, if it exists, or
	$\bar{\alpha} = 0$ otherwise.
\end{theorem}
Since $\norm*{f}_{k_\eta}^2 = c^\top K_\eta c$ (see the proof of Theorem~\ref{thm:representer:stable_c}), Equation~\eqref{eq:vec_c:stable} implies that the constraint on the RKHS norm in~\eqref{eq:opt_c:stable} is potentially translated to a lower bound on the regularization parameter.
This is in line with the typical effect of increasing the regularization parameter, which forces the optimal predictor to have a lower norm.

In Sections~\ref{sec:stability} and~\ref{sec:δstability}, we show that, by defining a suitable viability set $\Omega_k$, the model obtained using the proposed procedure is guaranteed to be BIBS stable, ISS, $\delta$BIBS stable and/or $\delta$ISS.
To this aim, we follow a rationale similar to~\cite{vanWaarde2022a,vanWaarde2023a}, also used in the preliminary conference version of this article~\cite{Scandella2023b}:
We first derive a sufficient condition on the predictor $f$ guaranteeing that $\Sigma(f)$ has the desired stability property.
Then, we define the viability set $\Omega_k$ to ensure that, for each $\eta\in\Omega_k$, the predictor $f$, obtained by solving~\eqref{eq:hp_sel_gen:stable} and~\eqref{eq:opt_c:stable}, satisfies such a sufficient condition.
We stress that, as in the classic kernel method described in Section~\ref{sec:method:classic}, also Problems~\eqref{eq:hp_sel_gen:stable} and~\eqref{eq:opt_c:stable} may give rise to a bilevel optimization problem with the consequent complications, since $J$ presents the same dependency on the dataset and/or on the solution of~\eqref{eq:opt_c:stable} as before.
Nevertheless, we underline that the viability set $\Omega_k$ does not depend on the dataset.
Therefore, if~\eqref{eq:hp_sel_gen:stable} is only solved for feasibility and not optimality, Problems~\eqref{eq:hp_sel_gen:stable} and~\eqref{eq:opt_c:stable} can be solved once and sequentially.

\section{Learning BIBS stable and ISS Models}
\label{sec:stability}

This section shows how the methodology proposed in Section~\ref{sec:method:stable} can be used to guarantee that the learned predictor is BIBS stable, ISS, or has a \emph{weak Asymptotic Gain} (wAG) property.

\subsection{Stability notions and results}
\label{sec:stability:notions}

The notions of BIBS stability, ISS, and wAG are formalized hereafter.
\begin{definition}[\cite{Andriano1997a}] \label{def:BIBS}
	System $\Sigma(f)$ is said to be \emph{BIBS stable} if for every $q \in [0,\infty)$, there exists $w \in [0,\infty)$, such that every solution $(x, v, \varepsilon)$ of $\Sigma(f)$ satisfies
	\begin{equation*}
		\norm*{x_0} \le q, \norm*{v}_\infty \le q, \norm*{\varepsilon}_\infty \le q
		\Rightarrow \norm*{x}_\infty \le w .
	\end{equation*}
\end{definition}
\begin{definition}[\cite{sontag_smooth_1989,jiang_input--state_2001}] \label{def:ISS}
	System $\Sigma(f)$ is said to be \emph{ISS} if there exist a class-$\mathcal{KL}$ function $\alpha$ and two class-$\mathcal{K}$ functions $\kappa_{\mathrm{u}}$ and $\kappa_{\mathrm{e}}$ such that every solution $(x, v, \varepsilon)$ of $\Sigma(f)$ satisfies
	\begin{equation*}
		\forall t \in \N,
		\hspace{1ex}
		\norm*{x_t}
		\le \alpha \left( \norm*{x_0}, t \right)
		+ \kappa_{\mathrm{u}} \left( \norm*{v_{0:t}}_\infty \right)
		+ \kappa_{\mathrm{e}} \left( \norm*{\varepsilon_{0:t}}_\infty \right).
	\end{equation*}
\end{definition}
\begin{definition} \label{def:wAG}
	System $\Sigma(f)$ is said to have the \emph{wAG} property if there exist $b \in [0, \infty)$ and a class-$\mathcal{K}$ function $\kappa$ such that every solution $(x, v, \varepsilon)$ of $\Sigma(f)$ satisfies
	\begin{equation*}
		\limsup_{t\to\infty} \norm*{x_t} \le b + \kappa \left(\limsup_{t\to\infty} \max( \norm*{v_t}, \norm*{\varepsilon_t} ) \right).
	\end{equation*}
	In particular, $b$ and $\kappa$ are called the \emph{asymptotic bias} and \emph{asymptotic gain}, respectively.
	If $b=0$, we say that $\Sigma(f)$ has the \emph{asymptotic gain (AG)} property~\cite{Sontag1996a}.
\end{definition}
The following proposition establishes some sufficient conditions on the predictor $f$ guaranteeing that $\Sigma(f)$ has one or more of the considered stability properties.
\begin{proposition} \label{prop:stability:f} %
	Suppose that there exist $\mu \in [0,1)$, $\rho \in [0,\infty)$, $d \in [0,\infty)$, and a class-$\mathcal{K}$ function $\omega$ such that, for every $x\in\R^{2m}$ and $v\in\R$, the following conditions hold
	\begin{subequations} \label{eq:stability:cond}
		\begin{align}
			\norm*{x}^2 + \norm*{v}^2 \ge \rho
			&
			\Rightarrow \norm*{f(x,v)}^2 \le \frac{\mu}{m} \norm*{x}^2 + \omega(\norm*{v}),
			\label{eq:stability:cond:contractive}
			\\
			\norm*{x}^2 + \norm*{v}^2 < \rho
			&
			\Rightarrow \norm*{f(x,v)} \le d.
			\label{eq:stability:cond:boundedness}
		\end{align}
	\end{subequations}
	Then:
	\begin{enumerate}[label=\textbf{\Alph*.}, ref=\textbf{\Alph*}]
		\item \label{prop:stability:f:BIBS} $\Sigma(f)$ is BIBS stable.
		\item \label{prop:stability:f:wAG} If $f$ is continuous, $\Sigma(f)$ has the wAG property with asymptotic bias $b(\rho)$, where $b(\cdot)$ is a class-$\mathcal{K}$ function of $\rho$.
		\item \label{prop:stability:f:ISS} If~\eqref{eq:stability:cond} holds with $\rho = 0$, $\Sigma(f)$ is ISS.
	\end{enumerate}
\end{proposition}
The proof of Proposition~\ref{prop:stability:f} is omitted, since it follows similar arguments of the proof of Proposition~\ref{prop:δstability:f} stated in the next section.

\subsection{Learning with stability constraints}
\label{sec:stability:estimation}

To show that the procedure explained in Section~\ref{sec:method:stable} can be used to guarantee that the learned models are BIBS stable and/or ISS, we proceed in two steps:
\begin{enumerate}[label=\textbf{S\arabic*.},ref=\textbf{S\arabic*}]
	\item \label{proc:stability:step:1}
	We derive a condition on the kernel function ensuring that all the elements of the associated RKHS satisfy Conditions~\eqref{eq:stability:cond} with $\rho,\mu,d \in [0, \infty)$ and $\omega$ of class-$\mathcal{K}$.
	\item \label{proc:stability:step:2} We design the viability set $\Omega_k$ in order to guarantee that the predictor $f$ selected by solving~\eqref{eq:opt_c:stable} satisfies Conditions~\eqref{eq:stability:cond} with $\mu\in[0,1)$.
\end{enumerate}
Steps~\ref{proc:stability:step:1} and~\ref{proc:stability:step:2} guarantee that every feasible predictor satisfies the assumptions of Proposition~\ref{prop:stability:f} thus yielding a stable $\Sigma(f)$.
\begin{remark}\label{rem:k_stability}
	We stress that satisfaction of the assumptions of Proposition~\ref{prop:stability:f} cannot be achieved in one step by just shaping the kernel structure $k$.
	In particular, the only RKHS whose functions all satisfy~\eqref{eq:stability:cond:contractive} is the trivial space, which only contains the $0$ function.
	Indeed, if for some $\eta\in\Phi_k$, $f\in\mathcal{H}(k_\eta)\setminus\{0\}$ satisfies~\eqref{eq:stability:cond:contractive} with $\mu \in (0,1)$, then $\mu^{-1} f \in\mathcal{H}(k_\eta)$ may not.
\end{remark}
Proposition~\ref{prop:stability:kernel} (proved in Appendix~\ref{prop:stability:kernel::proof}) fulfills Step~\ref{proc:stability:step:1} by providing sufficient conditions on a kernel function $k_\eta$ for which~\eqref{eq:stability:cond} holds for all functions $f\in\mathcal{H}(k_\eta)$, although with $\mu$ possibly larger than $1$.
\begin{proposition} \label{prop:stability:kernel} %
	Let $k$ be a kernel structure and $\eta \in \Phi_k$.
	Assume that there exist $\nu ,s\in [0, \infty)$ such that
	\begin{subequations} \label{eq:stability:kernel:cond}
		\begin{align}
			\forall a \in \R^{2m+1},
			& & \norm*{a}^2 \ge \nu
			& \Rightarrow k_\eta(a,a) \le \norm*{a}^2,
			\label{eq:stability:kernel:cond:contractive}
			\\
			\forall a \in \R^{2m+1},
			& & \norm*{a}^2 < \nu
			& \Rightarrow k_\eta(a,a) \le s.
			\label{eq:stability:kernel:cond:boundedness}
		\end{align}
	\end{subequations}
	Then, every function $f \in \mathcal{H}(k_\eta)$ satisfies conditions~\eqref{eq:stability:cond} with $\mu = m \norm*{f}_{k_\eta}^2$, $\rho = \nu$, $d = \sqrt{s} \norm*{f}_{k_\eta}$ and $\omega(p) = \norm*{f}_{k_\eta}^2 p^2$.
\end{proposition}
Let $k_\eta$ be a kernel function satisfying the assumptions of Proposition~\ref{prop:stability:kernel}.
Then, each $f\in\mathcal{H}(k_\eta)$ satisfies~\eqref{eq:stability:cond} with $\mu = m \norm*{f}_{k_\eta}^2$.
Since the constraint $m \norm{f}_{k_\eta}^2 \le \chi < 1$ is already present in~\eqref{eq:opt_c:stable}, then Step~\ref{proc:stability:step:2} is fulfilled if we can properly define the viability set $\Omega_k$ in such a way that all the solutions of~\eqref{eq:hp_sel_gen:stable} satisfy the assumptions of Proposition~\ref{prop:stability:kernel}.
This leads to the following definition.
\begin{definition} \label{def:ρ-viability}
	Let $k$ be a kernel structure.
	Given $\rho \in [0, \infty]$, we define the \emph{$\rho$-viability set} of $k$ as
	\begin{equation*}
		\Theta^\rho_k \coloneq
		\big\{
			\eta \in \Phi_k \st
			\exists \nu \in [0, \rho] \cap \R,\, \exists s \in [0, \infty),\, \textnormal{\eqref{eq:stability:kernel:cond} holds}
		\big\}.
	\end{equation*}
	If $\Theta^\rho_k \ne \varnothing$, we say that the kernel structure $k$ is \emph{$\rho$-viable}.
\end{definition}
Finally, joining Steps~\ref{proc:stability:step:1} and~\ref{proc:stability:step:2} and using Proposition~\ref{prop:stability:f} directly leads to Theorem~\ref{thm:BIBS-ISS} below, which is the main result of the present section.
\begin{theorem} \label{thm:BIBS-ISS}
	Let $\rho \in [0,\infty]$, $k$ be a $\rho$-viable kernel structure, and $(\beta, \eta)$ be an optimal solution of~\eqref{eq:hp_sel_gen:stable} obtained with $\Omega_k = \Theta^\rho_k$.
	Let $f$ be the optimal solution of~\eqref{eq:opt_c:stable}.
	Then, $\Sigma(f)$ is BIBS stable.
	If, in addition, $f$ is continuous, then $\Sigma(f)$ has the wAG property with asymptotic bias $b(\rho)$, where $b(\cdot)$ is a class-$\mathcal{K}$ function of $\rho$.
	Finally, if $\rho=0$, then $\Sigma(f)$ is ISS.
\end{theorem}
It is important to observe that a solution of~\eqref{eq:hp_sel_gen:stable} exists only if $\Theta^\rho_k \ne \varnothing$, namely, if $k$ is $\rho$-viable for the chosen $\rho\in[0,\infty]$.
$\rho$-viability of some the most commonly used kernels is discussed in Section~\ref{sec:kernels}, where we also characterize their $\rho$-viability set.

We observe that $\infty$-viability is sufficient for BIBS stability, while $0$-viability is needed for ISS.
Since $\rho\mapsto \Theta^{\rho}_k$ is increasing, then achieving BIBS stability is easier than ISS, which matches with the fact that ISS is a stricter stability condition.
Moreover, $\rho$-viability for $\rho\in(0,\infty)$ can be interpreted as an ``intermediate'' stability property between ISS and BIBS stability.
Finally, we underline that Theorem~\ref{thm:BIBS-ISS} also states that $\Sigma(f)$ yields the wAG property if $k$ is $\rho$-viable for some $\rho$ and $f$ is continuous.
In this respect, we point out that, since $f$ is a linear combination of representer functions, $f$ is continuous when the functions $k(z, \cdot)$ are continuous for each $z \in \R^{2m+1}$.
\begin{remark} \label{rem:iss:BIBS}
	Regarding the BIBS stability requirement, we stress that, in principle, BIBS stability may be obtained whenever the predictor is a bounded function. 
	Hence, every bounded kernel produces a BIBS stable system~\cite[Prop.~11]{Pillonetto2018a}.
	However, the bound achievable in this way corresponds to a limit case of Proposition~\ref{prop:stability:f} in which $\rho=\infty$.
	As clear from the proof of the proposition, this gives no information on how the amplitude of the input affects that of the output.
	The result of Theorem~\ref{thm:BIBS-ISS}, instead, provides a bound composed of two terms, the first depending on $\rho$ and independent of the input, and the second given by a class-$\mathcal{K}$ function of the norm of the input.
	Therefore, for $\rho=0$, Theorem~\ref{thm:BIBS-ISS} also provides a characterization of \emph{continuity} of the input-output operator associated with the predictor.
	For $\rho>0$, a bias is introduced scaling with $\rho$, thereby providing a combined effect of continuity and uniform boundedness.
\end{remark}
\begin{remark} \label{rem:iss:equilibrium}
	$0$-viability, needed for ISS, implies $f(0_{2m+1}) = 0$.
	Conversely, ISS implies that $(x, v, \varepsilon)$, with $x_t = 0_{2m}$ and $\varepsilon_t = v_t = 0$ for all $t \in \N$, is a solution of $\Sigma(f)$, which in turn implies $f(0_{2m+1}) = 0$.
	Therefore, Theorem~\ref{thm:BIBS-ISS} only characterizes ISS with respect to the origin.
	This property is also reflected in Condition~\eqref{eq:stability:cond:contractive} that implies $k(0_{2m+1},0_{2m+1}) = 0$ and, therefore, $f(0_{2m+1}) = 0$ for each $f \in \mathcal{H}(k_\eta)$~\cite[Prop.~2.3]{Saitoh2016a}.
	Imposing ISS with respect to a different \emph{known} equilibrium pair $(x^\star,v^\star)$ is possible by applying the preliminary data transformation $(x,v) \mapsto (x-x^\star, v-v^\star)$ before using Theorem~\ref{thm:BIBS-ISS}.
\end{remark}

\section{Learning $\delta$BIBS stable and $\delta$ISS Models}
\label{sec:δstability}

This section shows how the methodology proposed in Section~\ref{sec:method:stable} can be used to guarantee that the learned predictor is $\delta$BIBS stable, $\delta$ISS, and/or has a \emph{weak incremental asymptotic gain} (w$\delta$AG) property.
We follow the same line of reasoning of Section~\ref{sec:stability} and provide the ``incremental version'' of the results therein.

\subsection{Incremental stability notions and results}
\label{sec:δstability:notions}

Consider two solutions $(x^a,v^a,\varepsilon^a)$ and $(x^b,v^b,\varepsilon^b)$ of $\Sigma(f)$ and let $\tilde{x} \coloneq x^a - x^b$, $\tilde{v} \coloneq v^a - v^b$, and $\tilde{\varepsilon} \coloneq \varepsilon^a - \varepsilon^b$.
The notions of $\delta$BIBS stability, $\delta$ISS, and w$\delta$AG are formalized in the following definitions.
\begin{definition} \label{def:δBIBS}
	System $\Sigma(f)$ is said to be \emph{$\delta$BIBS stable} if, for every $q \in [0,\infty)$, there exists $w \in [0,\infty)$, such that every two solutions $(x^a,v^a,\varepsilon^a)$ and $(x^b,v^b,\varepsilon^b)$ of $\Sigma(f)$ satisfy
	\begin{equation*}
		\norm*{\tilde{x}_0} \le q, \norm*{\tilde{v}}_\infty \le q, \norm*{\tilde{\varepsilon}}_\infty \le q
		\Rightarrow \norm*{\tilde{x}}_\infty \le w.
	\end{equation*}
\end{definition} 
\begin{definition}[\cite{Angeli2002a}] \label{def:δISS}
	System $\Sigma(f)$ is said to be \emph{$\delta$ISS} if there exist a class-$\mathcal{KL}$ function $\alpha$ and two class-$\mathcal{K}$ functions $\kappa_{\mathrm{u}}$ and $\kappa_{\mathrm{e}}$ such that every two solutions $(x^a,v^a,\varepsilon^a)$ and $(x^b,v^b,\varepsilon^b)$ of $\Sigma(f)$ satisfy
	\begin{equation*}
		\forall t \in \N,
		\hspace{1ex}
		\norm*{\tilde{x}_t}
		\le \alpha \left( \norm*{\tilde{x}_0}, t \right)
		+ \kappa_{\mathrm{u}} \left( \norm*{\tilde{v}_{0:t}}_\infty \right)
		+ \kappa_{\mathrm{e}} \left( \norm*{\tilde{\varepsilon}_{0:t}}_\infty \right).
	\end{equation*}
\end{definition}
\begin{definition} \label{def:δwAG}
	System $\Sigma(f)$ is said to have \emph{w$\delta$AG} property if there exist $b \in [0, \infty)$ and a class-$\mathcal{K}$ function $\kappa$ such that every two solutions $(x^a,v^a,\varepsilon^a)$ and $(x^b,v^b,\varepsilon^b)$ of $\Sigma(f)$
	satisfy
	\begin{equation*}
		\limsup_{t\to\infty} \norm*{\tilde{x}_t}
		\le b + \kappa\left(\limsup_{t\to\infty} \max\left(\norm*{\tilde{v}_t}, \norm*{\tilde{\varepsilon}_t}\right)\right).
	\end{equation*}
	In particular, $b$ and $\kappa$ are called \emph{incremental asymptotic bias} and \emph{incremental asymptotic gain}, respectively.
	If $b=0$, we say that $\Sigma(f)$ has the \emph{incremental asymptotic gain ($\delta$AG)} property.
\end{definition}
Proposition~\ref{prop:δstability:f} (proved in Appendix~\ref{prop:δstability:f::proof}) stated below parallels Proposition~\ref{prop:stability:f} by giving sufficient conditions on the predictor $f$ guaranteeing that $\Sigma(f)$ is $\delta$BIBS stable, $\delta$ISS, and has the w$\delta$AG property.
\begin{proposition} \label{prop:δstability:f} %
	Suppose that there exist $\mu \in [0,1)$, $\rho \in [0,\infty)$, $d \in [0,\infty)$, and a class-$\mathcal{K}$ function $\omega$ such that, for all $x^a,x^b\in\R^{2m}$ and $v^a,v^b\in\R$, the following hold
	\begin{subequations} \label{eq:δstability:cond}
		\begin{align}
			\norm{\tilde{x}}^2 + \norm{\tilde{v}}^2 \ge \rho
			& \Rightarrow
			\norm{\tilde{f}}^2 \le \frac{\mu}{m} \norm{\tilde{x}}^2 + \omega(\norm{\tilde{v}}),
			\label{eq:δstability:cond:contractive}
			\\
			\norm{\tilde{x}}^2 + \norm{\tilde{v}}^2 < \rho
			& \Rightarrow \norm{\tilde{f}} \le d,
			\label{eq:δstability:cond:boundedness}
		\end{align}
	\end{subequations}
	where $\tilde{x} = x^a - x^b$, $\tilde{v} = v^a - v^b$, and $\tilde{f} = f(x^a,\tilde{v}^a) - f(x^b,\tilde{v}^b)$.
	Then:
	\begin{enumerate}[label=\textbf{\Alph*.}, ref=\textbf{\Alph*}]
		\item \label{prop:δstability:f:BIBS} $\Sigma(f)$ is $\delta$BIBS stable.
		\item \label{prop:δstability:f:wAG} If $f$ is uniformly continuous, $\Sigma(f)$ has the w$\delta$AG property with asymptotic incremental bias $b(\rho)$, where $b(\cdot)$ is a class-$\mathcal{K}$ function of $\rho$.
		\item \label{prop:δstability:f:ISS} If~\eqref{eq:δstability:cond} holds with $\rho = 0$, $\Sigma(f)$ is $\delta$ISS.
	\end{enumerate}
\end{proposition}

\subsection{Learning with incremental stability constraints}
\label{sec:δstability:estimation}

To show that the procedure explained in Section~\ref{sec:method:stable} can be used to guarantee that the learned models are $\delta$BIBS stable and/or $\delta$ISS, we follow the same line of reasoning of Section~\ref{sec:stability:estimation}, and we proceed in two steps (cf.~\ref{proc:stability:step:1}-\ref{proc:stability:step:2}):
\begin{enumerate}[label=\textbf{$\delta$S\arabic*.},ref=\textbf{$\delta$S\arabic*}]
	\item \label{proc:δstability:step:1} We derive a condition on the kernel function ensuring that all the elements of the associated RKHS satisfy Conditions~\eqref{eq:δstability:cond} with $\rho,\mu,d \in [0, \infty)$ and $\omega$ of class-$\mathcal{K}$.
	\item \label{proc:δstability:step:2} We design the viability set $\Omega_k$ in order to guarantee the predictor $f$ selected by solving~\eqref{eq:opt_c:stable} satisfies Conditions~\eqref{eq:δstability:cond} with $\mu\in[0,1)$.
\end{enumerate}
Step~\ref{proc:δstability:step:1} is fulfilled by the following proposition, which is analogous to Proposition~\ref{prop:stability:kernel}.
\begin{proposition} \label{prop:δstability:kernel} %
	Let $k$ be a kernel structure and $\eta \in \Phi_k$.
	Assume that there exist $\nu ,s\in [0, \infty)$ such that
	\begin{subequations} \label{eq:δstability:kernel:cond}
		\begin{align}
			\forall a,b \in \R^{2m+1},
			& & \norm*{a-b}^2 \ge \nu
			& \Rightarrow h_\eta(a,b) \le \norm*{a-b}^2,
			\label{eq:δstability:kernel:cond:contractive}
			\\
			\forall a,b \in \R^{2m+1},
			& & \norm*{a-b}^2 < \nu
			& \Rightarrow h_\eta(a,b) \le s,
			\label{eq:δstability:kernel:cond:boundedness}
		\end{align}
	\end{subequations}
	with $h_\eta(a,b) \coloneq k_\eta(a,a) - 2k_\eta(a,b) + k_\eta(b, b)$.
	Then, every function $f \in \mathcal{H}(k_\eta)$ satisfies Conditions~\eqref{eq:δstability:cond} with $\mu = m| f |_{k_\eta}^2$, $\rho = \nu$, $d = \sqrt{s} \norm*{f}_{k_\eta}$, and $\omega(p) = \norm*{f}_{k_\eta}^2 p^2$.
\end{proposition}
The function $h_\eta$ defined in Proposition~\ref{prop:δstability:kernel} is the squared kernel metric of $k_\eta$~\cite[Sec.~4.2]{Steinwart2008a}.
Hence, Condition~\eqref{eq:δstability:kernel:cond} in equivalent to asking that the canonical feature map of $k_\eta$ is non-expansive for points whose distance is large, and incrementally bounded for points that are closer.

As before, Step~\ref{proc:δstability:step:2} is fulfilled if we can properly define the viability set $\Omega_k$ in such a way that all the solutions of~\eqref{eq:hp_sel_gen:stable} satisfy the assumptions of Proposition~\ref{prop:δstability:kernel}.
This leads to the following definition (cf. Definition~\ref{def:ρ-viability}).
\begin{definition} \label{def:ρ-δviability}
	Let $k$ be a kernel structure.
	Given $\rho \in [0, \infty] $, we define the \emph{$\rho$-$\delta$viability set} of $k$ as
	\begin{equation*}
		\Delta^\rho_k \coloneq
		\left\{
			\eta \in \Phi_k \st
			\exists \nu \in [0, \rho] \cap \R,\, \exists s \in [0, \infty),\, \textnormal{\eqref{eq:δstability:kernel:cond} holds}
		\right\}.
	\end{equation*}
	If $\Delta^\rho_k \ne \varnothing$, we say that the kernel structure $k$ is \emph{$\rho$-$\delta$viable}.
\end{definition}
The $\rho$-$\delta$viability of some of the most commonly used kernels is discussed in Section~\ref{sec:kernels} along with their viability sets.
Joining~\ref{proc:δstability:step:1} and~\ref{proc:δstability:step:2} and using Proposition~\ref{prop:δstability:kernel} directly leads to Theorem~\ref{thm:δBIBS-δISS} below (analogous to Theorem~\ref{thm:BIBS-ISS}).
\begin{theorem}\label{thm:δBIBS-δISS}
	Let $\rho \in [0,\infty]$, $k$ be a $\rho$-$\delta$viable kernel structure, and $(\beta, \eta)$ be an optimal solution of~\eqref{eq:hp_sel_gen:stable} obtained with $\Omega_k = \Delta^\rho_k$.
	Let $f$ be the optimal solution of~\eqref{eq:opt_c:stable}.
	Then, $\Sigma(f)$ is $\delta$BIBS stable.
	If, in addition, $f$ is uniformly continuous, then $\Sigma(f)$ has the w$\delta$AG property with asymptotic incremental bias $b(\rho)$, where $b(\cdot)$ is a class-$\mathcal{K}$ function of $\rho$.
	Finally, if $\rho=0$, then $\Sigma(f)$ is $\delta$ISS.
\end{theorem}
Since $\rho \mapsto \Delta_k^\rho$ is increasing, the same relationship among the conditions for BIBS stability and ISS discussed in Section~\ref{sec:stability:estimation} holds for their incremental version as well.
Moreover, since $f$ is a linear combination of kernel functions, it is uniformly continuous when so are the functions $k_\eta^z$ for each $z \in \R^{2m+1}$.
Finally, we underline that similar conclusions of those given in Remark~\ref{rem:iss:BIBS} apply to $\delta$BIBS stability as well.

\section{Kernel Viability}
\label{sec:kernels}

In this section, we consider some of the most commonly used kernels, and we derive their $\rho$-viability and $\rho$-$\delta$viability sets.
Before delving into specific kernels, it is possible to state two general results and draw some preliminary considerations.
Specifically, Proposition~\ref{prop:k:cont-uni_cont} below states that, if the kernel is continuous (resp. uniformly continuous), Condition~\eqref{eq:stability:kernel:cond:boundedness} (resp. Condition~\eqref{eq:δstability:kernel:cond:boundedness}) is always satisfied and $\rho$-viability only depends on Condition~\eqref{eq:stability:kernel:cond:contractive} (resp. the $\rho$-$\delta$viability only depends on~\eqref{eq:δstability:kernel:cond:contractive}).
The proof of Proposition~\ref{prop:k:cont-uni_cont} directly follows by the fact that continuous functions are locally bounded, and a similar property holds for uniformly continuous functions (see Lemma~\ref{lem:f_unif_continuous_incr_bounded} in Appendix~\ref{prop:δstability:f::proof:unif-cont}).
Proposition~\ref{prop:k:bounded}, instead, states that all bounded kernels are $\infty$-viable and $\infty$-$\delta$viable.
\begin{proposition}[Continuous kernels] \label{prop:k:cont-uni_cont}
	Let $k$ be a kernel structure such that, for all $\eta \in \Phi_k$, the kernel function $k_\eta$ is continuous (resp. uniformly continuous).
	Then, for each $\rho \in [0, \infty]$, it holds that
	$\Theta^\rho_k =
	\left\{
	\eta \in \Phi_k \st
	\exists \nu \in [0, \rho] \cap \R, \textnormal{~\eqref{eq:stability:kernel:cond:contractive} holds}
	\right\}$
	(resp. $\Delta^\rho_k =
	\left\{
	\eta \in \Phi_k \st
	\exists \nu \in [0, \rho] \cap \R, \textnormal{~\eqref{eq:δstability:kernel:cond:contractive} holds}
	\right\}$).
\end{proposition}
\begin{proposition} [Bounded kernels] \label{prop:k:bounded}
	Let $k$ be a kernel structure such that, for each $\eta \in \Phi_k$, $k_\eta$ is bounded.
	Then, $k$ is $\infty$-viable and $\infty$-$\delta$viable with $\Theta^\infty_k = \Delta^\infty_k = \Phi_k$.
\end{proposition}
\begin{proof} \label{prop:k:bounded::proof}
	Let $b_\eta \in (0,\infty)$ be such that $\norm*{k_\eta (a,b)} \le b_\eta$ for all $a,b \in \R^{2m+1}$.
	Condition~\eqref{eq:stability:kernel:cond:boundedness} trivially holds with $\nu = s = b_\eta$.
	Similarly, also~\eqref{eq:stability:kernel:cond:contractive} holds with $\nu = b_\eta$ since $\norm*{a}^2 \ge c_\eta \Rightarrow \norm*{a}^2 \ge k_\eta(a,a)$ for all $a\in\R^{2m+1}$.
	Hence, $\Theta^\infty_k = \Phi_k$.
	Finally, $\Delta^\infty_k = \Phi_k$ can be proved similarly by noticing that, for every $a,b \in \R^{2m+1}$, $h_\eta(a,b) \le 4 b_\eta$.
\end{proof}

\subsection{Degenerate kernels}
\label{sec:kernels:degenerate}

This section considers a \emph{degenerate} kernel structure $k$~\cite[Def.~4.1]{Rasmussen2006a}, namely, such that, for each $\eta \in \Phi_k$, there exist $p_\eta \in \N$ and $\Gamma_\eta : \R^{2m+1} \to \R^{p_\eta}$, such that $k_\eta(a,b) = \Gamma_\eta(a)^\top \Gamma_\eta(b)$ for every $a, b \in \R^{2m+1}$.
Since $k_\eta(a,a) = \norm*{\Gamma_\eta(a)}^2$ and $h_\eta(a,b) = \norm*{\Gamma_\eta(a) - \Gamma_\eta(b)}^2$, Conditions~\eqref{eq:stability:kernel:cond} and~\eqref{eq:δstability:kernel:cond} can be easily rewritten in terms of $\Gamma_\eta$.
Moreover, $k_\eta$ is continuous or bounded if $\Gamma_\eta$ enjoys the same property.
Therefore, Proposition~\ref{prop:k:bounded} and the first statement of Proposition~\ref{prop:k:cont-uni_cont} can easily be rewritten in terms of $\Gamma_\eta$ instead of $k_\eta$.

A relevant degenerate kernel is the linear affine kernel analyzed in Proposition~\ref{prop:k:linear} below (proved in Appendix~\ref{prop:k:linear::proof}).
It defines the space of all linear affine functions.
Hence, it produces linear models.
Although there are more specialized methods for kernel-based identification of linear systems~\cite{Pillonetto2014a}, Proposition~\ref{prop:k:linear} shows that, for $0$-viability, the hyperparameter that regulates the affine part of the kernel needs to be $0$.
In this way, we guarantee that the space only contains linear functions.
Furthermore, the $0$-$\delta$viability set, the $\infty$-$\delta$viability set, and the $\infty$-viability are the same.
This is consistent with the fact that $\delta$ISS, $\delta$BIBS stability, and BIBS stability are all equivalent notions for linear systems.
\begin{proposition}[Linear affine kernel] \label{prop:k:linear} %
	Let $k$ be a kernel structure such that $\Phi_k = [0, \infty)^2$ and $k_{\eta}(a,b) = \tau a^\top b + \sigma$ for all $a,b \in \R^{2m+1}$ and $\eta = (\tau, \sigma) \in \Phi_k$.
	Then, for each $\rho \in [0, \infty]$, it follows that $\Theta^\rho_k = \{ (\tau, \sigma) \in \Phi_k \st \tau \in [0,1] \land \sigma \in [0, \rho(1-\tau)] \}$ and $\Delta_k^\rho = [0,1] \times [0,\infty)$.
\end{proposition}

Another common degenerate kernel is the polynomial kernel analyzed in Proposition~\ref{prop:k:polynomial} below (proved in Appendix~\ref{prop:k:polynomial::proof}).
This kernel is widely used in the literature~\cite{Aguirre2002a,Karami2021a,Farina2012a} because it defines the space of polynomials with an arbitrary degree~\cite[Prop.~2.1]{Scholkopf2018a}.
Unfortunately, this kernel is only viable when it is reduced to the linear kernel previously analyzed.
\begin{proposition}[Polynomial kernel] \label{prop:k:polynomial} %
	Let $k$ be a kernel structure such that $\Phi_k = \{ 2, \cdots \} \subseteq \N$ and $k_{\eta}(a,b) = (a^\top b)^\eta$ for all $a,b \in \R^{2m+1}$ and $\eta \in \Phi_k$.
	Then, for every $\rho \in [0, \infty]$, $\Theta^\rho_k = \Delta^\rho_k = \varnothing$.
\end{proposition}

Although Proposition~\ref{prop:k:polynomial} gives a somewhat negative result, nonlinear polynomial kernels (as well as other degenerate kernels) can always be modified to enable their usage to learn stable systems.
This is briefly discussed in Section~\ref{sec:kernels:composition}.

\subsection{Stationary Kernels}
\label{sec:kernels:stationary}

Stationary kernels enjoy many relevant properties (see~\cite[Sec.~4.2.1]{Rasmussen2006a} or~\cite[Sec.~4.4]{Scholkopf2018a} for more details), and they include many widely-used kernels such as the \emph{Gaussian Kernel}~\cite{Steinwart2006a} and all \emph{isotropic kernels}.
\begin{definition}
	A kernel structure $k$ is said to be stationary if, for every $\eta \in \Phi_k$, there exists a function $\bar{k}_\eta:\R^{2m+1}\to\R$, such that $k_\eta(a,b) = \bar{k}_\eta(a-b)$ for all $a,b \in \R^{2m+1}$.
\end{definition}
The following result characterizes viability of stationary kernels.
\begin{theorem} [Stationary kernels] \label{thm:k:stationary} %
	Let $k$ be a stationary kernel structure.
	Then, for each $\rho \in [0, \infty]$,
	\begin{align*}
		\Theta^\rho_k
		&
		= \{ \eta \in \Phi_k \st \bar{k}_\eta(0_{2m+1}) \le \rho \} ,\\
		\Delta^\rho_k
		&
		= \big\{
			\eta \in \Phi_k
			\st
			\forall z \in \R^{2m+1},
			\\ & \hspace{18mm}
			\norm*{z}^2 \ge \rho
			\Rightarrow
			2\bar{k}_\eta(0_{2m+1}) - 2\bar{k}_\eta(z) \le \norm*{z}^2
		\big\}.
	\end{align*}
	Moreover, $\Delta^\rho_k \supseteq \{ \eta \in \Phi_k \st 4 \bar{k}_\eta(0_{2m+1}) \le \rho \}$.
\end{theorem}
Theorem~\ref{thm:k:stationary} is proved in Appendix~\ref{thm:k:stationary::proof}.
A first immediate consequence of Theorem~\ref{thm:k:stationary} is that all stationary kernels are always $\infty$-viable and $\infty$-$\delta$viable and, hence, they can be used to learn BIBS stable and $\delta$BIBS stable models without additional restrictions on the hyperparameters.
This is formalized by the following corollary.
\begin{corollary}
	Let $k$ be a stationary kernel structure.
	Then, $\Theta^\infty_k = \Delta^\infty_k = \Phi_k$.
\end{corollary}
A second relevant consequence of Theorem~\ref{thm:k:stationary} is that no non-trivial stationary kernel is $0$-viable.
Hence, stationary kernels are generally not suitable to learn ISS systems.
\begin{corollary} \label{cor:k:nonstationary_ISS}
	Let $k$ be a stationary kernel structure.
	Then, $\Theta^0_k = \{ \eta \in \Phi_k \st \forall a,b \in \R^{2m+1}, k_\eta(a,b) = 0 \}$, that is, no non-trivial stationary kernel is $0$-viable.
\end{corollary}
\begin{proof}
	The claim directly follows from Theorem~\ref{thm:k:stationary} since $\bar{k}_\eta(0_{2m+1}) \le 0$ if and only if $k_\eta(a,a)=0$ for all $a\in\R^{2m+1}$, which is true if and only if $k_\eta(a,b)=0$ for all $a,b\in\R^{2m+1}$ (see~\eqref{eq:barkz_le_bark0}).
\end{proof}
Two important stationary kernels are the Gaussian Kernel (which defines a RKHS dense in the space of continuous functions~\cite{Micchelli2006a}) and the Mat\'{e}rn Kernel~\cite{Stein1999a}.
They are analyzed in Propositions~\ref{prop:k:gaussian} and~\ref{prop:k:matern} below (proved in Appendix~\ref{prop:k:gaussian::proof} and Appendix~\ref{prop:k:matern::proof}, respectively).
Proposition~\ref{prop:k:matern} only considers the first degree Matérn Kernel for simplicity, but analogous results can be obtained with larger degrees.
Both kernels are $0$-$\delta$viable, whilst the hyperparameters need to meet certain conditions for $\delta$ISS.
\begin{proposition}[Gaussian Kernel] \label{prop:k:gaussian} %
	Let $k$ be a stationary kernel structure such that $\Phi_k = [0, \infty)^3$ and, for every $z \in \R^{2m+1}$ and $\eta = (\tau, \gamma, \sigma) \in \Phi_k$, $\bar{k}_\eta(z) = \tau \exp ( - \gamma \norm*{z}^2 ) + \sigma$.
	Then, for each $\rho \in [0, \infty]$, we have
	\begin{align*}
		\Theta^\rho_k & = \{ (\tau, \gamma, \sigma) \in \Phi_k \st \tau + \sigma \le \rho \}, \\
		\Delta^\rho_k & =
		\begin{cases}
			\{ (\tau, \gamma, \sigma) \in \Phi_k \st 2 \tau \gamma \le 1 \},
			& \textnormal{if } \rho = 0,\\
			\left\{
			(\tau, \gamma, \sigma) \in \Phi_k \st
			v(\tau, \gamma) \le \rho
			\right\}
			& \textnormal{if } \rho \in (0, \infty),\\
			\Phi_k,
			& \textnormal{if } \rho = \infty,
		\end{cases}
	\end{align*}
	in which $v(\tau, \gamma) = 2\tau + \gamma^{-1} W\left( -2\gamma \tau e^{-2 \gamma \tau} \right)$ and $W$ is the principal branch of the Lambert W function~\cite{Corless1996a}.
	Additionally, $\Delta^\rho_k \supseteq \left[0, \frac{\rho}{2} \right] \times [0, \infty)^2$.
\end{proposition}
\begin{proposition}[Mat\'{e}rn Kernel] \label{prop:k:matern} %
	Let $k$ be a stationary kernel structure such that $\Phi_k = [0, \infty)^3$ and
	$\bar{k}(z) = \tau \big( 1 + \gamma \sqrt{3} \norm*{z} \big) \exp \big(- \gamma \sqrt{3} \norm*{z} \big) + \sigma$,
	for all $z \in \R^{2m+1}$ and $\eta = (\tau, \gamma, \sigma) \in \Phi_k$.
	Then, for each $\rho \in [0, \infty]$, we have
	\begin{align*}
		\Theta^\rho_k &= \{ (\tau, \gamma, \sigma) \in \Phi_k \st \tau + \sigma \le \rho \} \, , \\
		\Delta^0_k &= \{ (\tau, \gamma, \sigma) \in \Phi_k \st 3 \tau \gamma^2 \le 1 \} \, .
	\end{align*}
\end{proposition}
Finally, Proposition~\ref{prop:k:pillo} below (proved in Appendix~\ref{prop:k:pillo::proof}) analyzes the kernel proposed in~\cite{Pillonetto2011a}, which is specifically designed for the identification of predictors for NARX systems.
Loosely speaking, this kernel is designed so as older measurements weight less than recent ones in the computation of the predicted value (with the parameter $\xi$ playing the role of a forgetting factor).
Since this is a stationary kernel, according to Corollary~\ref{cor:k:nonstationary_ISS} it is not suitable for learning ISS systems, but it guarantees that the learned model is $\delta$BIBS and BIBS stable without constraining the hyperparameters' selection optimization problem, i.e.~\eqref{eq:hp_sel_gen:classic} and~\eqref{eq:hp_sel_gen:stable} are equivalent.
\begin{proposition} \label{prop:k:pillo} %
	Let $k$ be a stationary kernel structure such that $\Phi_k = [0, \infty)^3 \times \{1,\ldots,m\}$ and, for all $z = (y_1,\ldots,y_m, u_1,\ldots,u_{m+1}) \in \R^{2m+1}$ and $\eta = (\tau, \gamma, \sigma,p) \in \Phi_k$,
	\begin{equation*}
		\bar{k}_\eta(z) = \tau \sum_{t=0}^{m-p} \exp( -\xi t -\gamma \norm*{z_t}^2 ),
	\end{equation*}
	in which $z_t = (y_{t+1},\ldots,y_{t+p}, u_{t+1},\ldots,u_{t+p}) \in \R^{2p}$.
	Then
	\begin{align*}
		\Theta^0_k &= \{ (\tau, \gamma, \xi, p) \in \Phi_k \st \tau = 0 \} \, , \\
		\Delta^0_k &= \{(\tau, \gamma, \xi, p) \in \Phi_k \st 2 \gamma\tau \pi(\xi, p) \le 1 \} \, ,
	\end{align*}
	where $\pi(\xi, p)\coloneq	m - p + 1$ if $\xi=0$ and  $\pi(\xi, p)\coloneq\frac{ 1-e^{-(m-p+1)\xi} }{ 1-e^{-\xi} }$ if $\xi>0$.
\end{proposition}
We emphasize that Proposition~\ref{prop:k:pillo} validates the rationale used in~\cite{Pillonetto2011a} behind the construction of the kernel.
Firstly, notice that $\pi(\cdot, p)$ is a continuous non-increasing function in $[0, \infty)$ whose image is $(1, m - p + 1]$.
Therefore, a necessary condition for $0$-$\delta$viability is that $2 \gamma \tau \le 1$.
If this is the case, we need to select $\xi$ large enough to guarantee that $2 \gamma\tau \pi(\xi, p) \le 1$.
Thus, in qualitative terms, the predictor must lead to a system forgetting past initial conditions sufficiently fast.

\subsection{Composition of viable kernels}
\label{sec:kernels:composition}

It is well known that the sum of symmetric positive semidefinite kernels is a positive semidefinite kernel~\cite[Prop.~13.1]{Scholkopf2018a}.
Proposition~\ref{prop:k:sum} below (proved in Appendix~\ref{prop:k:sum::proof}) gives sufficient conditions under which the linear combination of viable kernels is again viable.
\begin{proposition}[Sum of kernels] \label{prop:k:sum} %
	Let $(w_1)_{i=1}^q$ be $q \in \N$ kernel structures and $k$ be the kernel structure such that $\Phi_k = (0,\infty)^q \times \Phi_{w_1} \times \cdots \times \Phi_{w_q}$ and
	\begin{align*}
		\forall a,b \in \R^{2m+1}, & & k_\eta (a, b) &\coloneq \sum_{i=1}^q \tau_i w_{i,\eta_i} (a, b) \, ,
	\end{align*}
	for all $\eta = ( \tau_1, \cdots, \tau_q, \eta_1 \cdots, \eta_q )\in \Phi_k$.
	Then, for each $\rho \in [0, \infty]$, the following inclusions hold
	\begin{align*}
		\Theta_k^\rho
		&
		\supseteq
		\left\{
			\big( (\tau_i)^q_{i=1}, (\eta_i)^q_{i=1}\big) \in (0, \infty)^q\times \Theta^\rho_w \st \sum^q_{i=1} \tau_i \le 1
		\right\} \, ,
		\\
		\Delta_k^\rho
		&
		\supseteq
		\left\{
			\big( (\tau_i)^q_{i=1}, (\eta_i)^q_{i=1}\big) \in (0, \infty)^q\times \Delta^\rho_w \st \sum^q_{i=1} \tau_i \le 1
		\right\} \, ,
	\end{align*}
	where $\Theta^\rho_w = \Theta^\rho_{w_1} \times \cdots \times \Theta^\rho_{w_q}$ and $\Delta^\rho_w = \Delta^\rho_{w_1} \times \cdots \times \Delta^\rho_{w_q}$.
\end{proposition}
As implied by Corollary~\ref{cor:k:nonstationary_ISS}, stationary kernels are not suitable for learning ISS systems.
To address this problem, Proposition~\ref{prop:k:prod_stationary} below (proved in Appendix~\ref{prop:k:prod_stationary::proof}) provides a way to modify a stationary kernel in order to obtain a non-trivial $0$-viable kernel structure, which is therefore compatible with the conditions for ISS.
\begin{proposition} \label{prop:k:prod_stationary} %
	Let $\ell$ and $w$ be kernel structures such that $w$ is stationary.
	Let $k$ be a kernel structure such that $\Phi_k = \Phi_\ell \times \Phi_w$ and $k_\eta(a,b) = \ell_{\eta_1}(a,b) w_{\eta_2}(a,b)$ for all $a,b \in \R^{2m+1}$ and all $\eta = (\eta_1, \eta_2) \in \Phi_k$.
	Then, for each $\rho \in [0, \infty]$, $\Theta_k^\rho \supseteq \left\{ (\eta_{1},\eta_2)\in\Theta_{\ell}^{\rho}\times\Phi_{w} \st \bar{w}_{\eta_{2}}(0_{2m+1}) \le 1 \right\}$.
\end{proposition}
An example of application of Proposition~\ref{prop:k:prod_stationary} can be easily constructed by taking $\ell$ as a degenerate kernel, and $w$ as one of the stationary kernels analyzed in Section~\ref{sec:kernels:stationary}.
In particular, let $p \in \N$ and $\Gamma : \mathbb{R}^{2m+1} \to \R^p$ be such that $\norm*{\Gamma(a)}^2 \le \norm*{a}^2$, for all $a \in \R^{2n+1}$.
Then, the kernel
\begin{equation} \label{eq:k:fea_gaussian}
	k_\eta(a,b) \coloneq \Gamma(a)^\top \Gamma(b) ( \tau \exp(- \gamma \norm*{a - b}^2) + \sigma)
\end{equation}
is $0$-viable with $\Theta^\rho_k \supseteq \{ (\tau, \gamma, \sigma) \in \Phi_k \st \tau + \sigma \le 1 \}$, where we have used the properties of the Gaussian Kernel analyzed in Proposition~\ref{prop:k:gaussian}.

\section{Numerical Experiments}
\label{sec:sim}

This section provides some numerical experiments validating the theoretical findings.
Section~\ref{sec:sim:toy} presents a Monte Carlo analysis on two models that satisfy the assumptions of Proposition~\ref{prop:stability:f} and~\ref{prop:δstability:f}, respectively.
This experiment allows us to validate the proposed method in the case where the dataset is generated by a system having the same structure of the identified one.
Section~\ref{sec:sim:hodgkin}, instead, reports a more realistic application concerning the identification of the potassium ion gate's dynamics in the \emph{Hodgkin-Huxley}'s neuron model~\cite{Hodgkin1952a}.

\subsection{Learning stable systems}
\label{sec:sim:toy}

We consider two discrete-time models $\mathsf{A}$ and $\mathsf{B}$ with input $u$ and output $y$ that obey the two difference equations
\begin{align*}
	\mathsf{A}: \ y_t & = 0.2 \norm*{p_t} \sqrt{\sin(\norm*{p_t})+1},
	&
	\mathsf{B}: \ y_t & = 0.2 \sin(\norm*{p_t})^2
\end{align*}
where $p_t\coloneq(y_{t-2},y_{t-1},u_{t-2},u_{t-1})$.
Using Propositions~\ref{prop:stability:f} and~\ref{prop:δstability:f}, we can show that System $\mathsf{A}$ is ISS and $\mathsf{B}$ is $\delta$ISS.
For the identification of model $i \in \{\mathsf{A}, \mathsf{B}\}$ we employ a dataset $\mathcal{D}_i$ containing $200$ samples of the input and output measurements.
In particular, for every $t \in \{1, \ldots, 200\}$, the pairs $(\bar{u}_t, \bar{y}_t) \in \mathcal{D}_i$ are obtained by sampling the random variables
\begin{align*}
	\bar{u}_t & = u_{t+2} \, ,
	&
	\bar{y}_t & = y_{t+2} + w^i_{t+2} \, ,
\end{align*}
where $w^{\mathsf{A}}_t \sim \mathcal{N}_0(0.05)$, $w^{\mathsf{B}}_t \sim \mathcal{N}_0(0.02)$ ($\mathcal{N}_0(S)$ denotes the $0$-mean Normal distribution with covariance matrix $S$); the signal-to-noise ratio is close to $10$ in both cases.
Furthermore, $y_0 \sim \mathcal{N}_0(1)$, $y_1 \sim \mathcal{N}_0(1)$, and the input satisfies $u_t \sim \mathcal{N}_0(1)$ for every $t \in \N$.
All the introduced random variables are mutually independent.
For every model $i \in \{\mathsf{A}, \mathsf{B}\}$, the performances are evaluated on a validation dataset $\mathcal{D}^v_i$ sampled from the same distribution of the dataset $\mathcal{D}_i$.

To better validate the model, we consider both the performance in $1$-step-ahead prediction and in simulation.
In particular, the $1$-step-ahead prediction sequence $\hat{y}^{\mathrm{pre}}$ and the simulated output sequence $\hat{y}$ are as defined in Section~\ref{sec:problem} using the inputs of the validation dataset $\mathcal{D}^v_i$.
Both sequences are equal to the output of the validation datasets for the first $m$ samples to guarantee a sensible initial condition for the simulation.
Prediction and simulation are evaluated in terms of their adherence to the output sequence $\bar{y}_t$.
In particular, we define the performance indexes
$q_{\mathrm{pre}} \coloneq \frac{1}{n-m} \sum_{t=m+1}^n \norm*{\bar{y}_t - \hat{y}^{\mathrm{pre}}_t}$
and
$q_{\mathrm{sim}} \coloneq \frac{1}{n-m} \sum_{t=m+1}^n \norm*{\bar{y}_t - \hat{y}_t}$.
The analysis is carried out using Monte Carlo simulations with $501$ runs.
For the learning procedure, we select $m=2$, $\iota=10^{-10}$, $\chi = 0.99$, and the hyperparameters are computed using the EB method, suitably modified to account for the stability constraint for the proposed method.
After selecting the hyperparameters, we employ the procedure presented in~\cite{Scandella2021a} to reduce the computational complexity of the estimated model by enforcing sparsity of the solution.

Since in the context of this validation we do know that $\mathsf{A}$ is ISS, the developed theory suggests using kernel~\eqref{eq:k:fea_gaussian} with $\Gamma(a) = a$ for all $a \in \R^{2m+1}$.
We thus consider the two predictors:
\begin{enumerate}[label={\textbf{A\alph*}.}, ref={\textbf{A\alph*}}]
	\item \label{en:sim:A:sta}
	$f$ obtained by solving~\eqref{eq:opt_c:classic}-\eqref{eq:hp_sel_gen:classic},
	\item \label{en:sim:A:iss}
	$f$ obtained by solving~\eqref{eq:hp_sel_gen:stable}-\eqref{eq:opt_c:stable} with $\Omega_k = \Theta_k^0$.
\end{enumerate}
Therefore,~\ref{en:sim:A:iss} guarantees that the estimated model is ISS while~\ref{en:sim:A:sta} does not.
As Figure~\ref{fig:perf_A} shows, imposing the right stability notion, ISS, yields similar performance to the standard predictor in both prediction and simulation.
However, the models estimated with~\ref{en:sim:A:iss} are ISS as proven in Theorem~\ref{thm:BIBS-ISS}, while~\ref{en:sim:A:sta} gives no guarantees of ISS.
Similar conclusions can also be obtained for model $\mathsf{B}$.
Since $\mathsf{B}$ is $\delta$ISS, the learning procedure is carried out using the Gaussian Kernel analyzed in Proposition~\ref{prop:k:gaussian}, and we consider two predictors:
\begin{enumerate}[label={\textbf{B\alph*}.}, ref={\textbf{B\alph*}}]
	\item \label{en:sim:B:sta}
	$f$ obtained by solving~\eqref{eq:opt_c:classic}-\eqref{eq:hp_sel_gen:classic},
	\item \label{en:sim:B:δiss}
	$f$ obtained by solving~\eqref{eq:hp_sel_gen:stable}-\eqref{eq:opt_c:stable} with $\Omega_k = \Delta_k^0$.
\end{enumerate}
In particular,~\ref{en:sim:B:δiss} guarantees that the estimated model is $\delta$ISS while~\ref{en:sim:B:sta} does not.
As shown in Figure~\ref{fig:perf_B}, we can draw the same observations as in the previously analyzed case.

These two examples show that, when the underlying system can be modeled by $\Sigma(f)$ for some predictor $f$ that satisfies the assumptions of the Propositions~\ref{prop:stability:f} and~\ref{prop:δstability:f}, the proposed method is able to guarantee that the estimated models have a desired stability property without significantly deteriorate the prediction and simulation capabilities.

\subsection{Potassium channel of an excitable cell}
\label{sec:sim:hodgkin}

We consider the $\mathrm{K^+}$ channel dynamics of the Hodgkin-Huxley's neuron model, which is given by the following differential equation~\cite[Eq.~(7,26)]{Hodgkin1952a}
\begin{equation*}
	\mathsf{H} :
	\left\{
		\begin{aligned}
			\dot{\kappa} & = \frac{(V+10) (1-\kappa)}{100 \left(e^{\frac{V}{10}+1}-1\right)}-\frac{e^{\frac{V}{80}}}{8} \kappa\\
			I & = 36 (V-12) \kappa^4
		\end{aligned}
	\right.
\end{equation*}
where $V$ is the input (the neuron's membrane potential) and $I$ is the output (the $\mathrm{K^+}$ current flowing across the membrane).
The objective is to learn a discrete-time predictor of the output $I$.
For the learning procedure, we consider a dataset $\mathcal{D}_{\mathsf{H}}$ generated by sampling (with sampling time $0.1\mathrm{s}$), for each $t \in \{1, \ldots, n \}$, the random signals
\begin{align*}
	\bar{u}_t & = V(49.9 + 0.1 t), &
	\bar{y}_t & = r(49.9 + 0.1 t, u, \kappa_0) \, ,
\end{align*}
where $n=201$, $\kappa_0 \sim \mathcal{N}_0 (1)$, $r(h, u, \kappa_0)$ is the forced output of system $\mathsf{H}$ at time $h$ given the initial condition $\kappa_0$, and $V(t)$ is chosen as the random variable
\begin{equation*}
	V(t) = \sum_{i=1}^{50} A_i \sin(2 \pi \nu_i t + \varphi_i),
\end{equation*}
where, for every $i \in \{1, \ldots, 50 \}$, $A_i \sim \mathcal{U}( 0.1, 0.5 )$, $\nu_i \sim \mathcal{U}( 0, 1 )$, $\varphi_i \sim \mathcal{U}( 0, 2\pi )$ ($\mathcal{U}(a,b)$ denotes the uniform distribution on $[a,b] \subseteq \R$).
Hence, $\bar{u}$ and $\bar{y}$ are measurements of $V$ and $I$, respectively.
The random variables $\{\kappa_0\} \cup \{A_i, \nu_i, \phi_i\}_{i=1}^{50}$ are all mutually independent.
For this experiment, we select $m=2$, $\iota=10^{-10}$, $\chi = 0.99$, and the hyperparameters are computed using the EB method modified to account for the stability constraint for the proposed method.
After selecting the hyperparameters, we employ the procedure presented in~\cite{Scandella2021a} to reduce the computational complexity of the estimated model by enforcing sparsity of the solution.
The identification algorithm is carried out using the Gaussian Kernel (described in Proposition~\ref{prop:k:gaussian}), and we consider three predictors:
\begin{enumerate}[label={\textbf{H\alph*}.}, ref={\textbf{H\alph*}}]
	\item \label{en:sim:H:sta}
	$f$ obtained by solving~\eqref{eq:opt_c:classic}-\eqref{eq:hp_sel_gen:classic},
	\item \label{en:sim:H:δbibs}
	$f$ obtained by solving~\eqref{eq:hp_sel_gen:stable}-\eqref{eq:opt_c:stable} with $\Omega_k = \Delta^\infty_k$,
	\item \label{en:sim:H:δiss}
	$f$ obtained by solving~\eqref{eq:hp_sel_gen:stable}-\eqref{eq:opt_c:stable} with $\Omega_k = \Delta^0_k$.
\end{enumerate}
The obtained predictors are evaluated using a validation dataset sampled from the same distribution from the dataset $\mathcal{D}_{\mathsf{H}}$ with $n = 5001$.
The analysis is carried out using Monte Carlo simulations with $501$ runs.
The performance in prediction and simulation of the estimated models are reported in Figures~\ref{fig:perf_hodgkin_pre} and~\ref{fig:perf_hodgkin_sim}, respectively.
We notice that~\ref{en:sim:H:sta} and~\ref{en:sim:H:δbibs} have similar prediction performance.
As explained in Remark~\ref{rem:iss:BIBS}, since the Gaussian kernel is bounded, the two methods generate BIBS stable models.
However,~\ref{en:sim:H:δbibs} also guarantees that the system is $\delta$BIBS and that it has the wAG and w$\delta$AG properties without performance deterioration.
Instead,~\ref{en:sim:H:δiss} causes a small decrease of predictive performance.
However, the estimated models are guaranteed to be $\delta$ISS.
Thus, as explained in Section~\ref{sec:problem}, the simulation error remains bounded by a class $\mathcal{K}$ function of the infinity norm of the past prediction error (see Definition~\ref{def:δISS}).
The beneficial effect of this property can be seen in the simulation performance shown in Figure~\ref{fig:perf_hodgkin_sim} where we can notice that the simulation error of~\ref{en:sim:H:δiss} remains of the same order of magnitude of the prediction error.
Instead,~\ref{en:sim:H:sta} presents a significantly larger simulation error that renders the models unsuitable for simulation purposes.
Similar conclusions can also be drawn from Figure~\ref{fig:ts_hodgkin} where it is possible to see that the prediction error is small for both models, but the simulation error of~\ref{en:sim:H:sta} is significantly larger in magnitude.

Finally, case~\ref{en:sim:H:δbibs} shows that only guaranteeing the weaker stability property of $\delta$BIBS with w$\delta$AG allows for better simulation performance without significantly impacting the prediction capability of the estimated model.
Thus, it provides a middle ground between the lack of guarantees of~\ref{en:sim:H:sta} and the strong ones of~\ref{en:sim:H:δiss}.

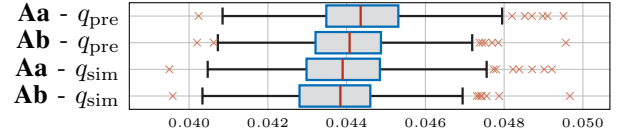
\begin{figure}[t]
	\centering
	\begin{tikzpicture}
		\begin{axis}[
			width=0.9\linewidth,
			height=3cm,
			ymin=0.5,
			ymax=4.5,
			boxplot/draw direction=x,
			xmajorgrids,
			ymajorgrids,
			xtick style = {draw=none},
			xticklabel style={
				font=\tiny,
				/pgf/number format/fixed,
				/pgf/number format/zerofill,
				/pgf/number format/precision=3,
			},
			scaled x ticks=false,
			ytick style = {draw=none},
			ytick={1,...,4},
			yticklabels={%
				\ref{en:sim:A:iss} - $q_{\mathrm{sim}}$,
				\ref{en:sim:A:sta} - $q_{\mathrm{sim}}$,
				\ref{en:sim:A:iss} - $q_{\mathrm{pre}}$,
				\ref{en:sim:A:sta} - $q_{\mathrm{pre}}$,
			},
		]
			\addplot+ [my boxplot] table [y=iss] {latex_data/bp_narx_iss_sim};
			\addplot+ [my boxplot] table [y=sta] {latex_data/bp_narx_iss_sim};
			\addplot+ [my boxplot] table [y=iss] {latex_data/bp_narx_iss_pre};
			\addplot+ [my boxplot] table [y=sta] {latex_data/bp_narx_iss_pre};
		\end{axis}
	\end{tikzpicture}
	\caption{Box plots of the performance indexes for system $\mathsf{A}$.}
	\label{fig:perf_A}
\end{figure}
\begin{figure}[t]
	\centering
	\begin{tikzpicture}
		\begin{axis}[
			width=0.9\linewidth,
			height=3cm,
			ymin=0.5,
			ymax=4.5,
			boxplot/draw direction=x,
			xmajorgrids,
			ymajorgrids,
			xtick style = {draw=none},
			xticklabel style={
				font=\tiny,
				/pgf/number format/fixed,
				/pgf/number format/zerofill,
				/pgf/number format/precision=3,
			},
			scaled x ticks=false,
			ytick style = {draw=none},
			ytick={1,...,4},
			yticklabels={%
				\ref{en:sim:B:δiss} - $q_{\mathrm{sim}}$,
				\ref{en:sim:B:sta} - $q_{\mathrm{sim}}$,
				\ref{en:sim:B:δiss} - $q_{\mathrm{pre}}$,
				\ref{en:sim:B:sta} - $q_{\mathrm{pre}}$,
			},
		]
			\addplot+ [my boxplot] table [y=diss] {latex_data/bp_narx_diss_sim};
			\addplot+ [my boxplot] table [y=sta] {latex_data/bp_narx_diss_sim};
			\addplot+ [my boxplot] table [y=diss] {latex_data/bp_narx_diss_pre};
			\addplot+ [my boxplot] table [y=sta] {latex_data/bp_narx_diss_pre};
		\end{axis}
	\end{tikzpicture}
	\caption{Box plots of the performance indexes for system $\mathsf{B}$.}
	\label{fig:perf_B}
\end{figure}
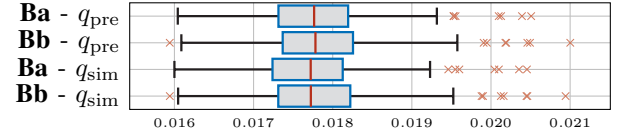
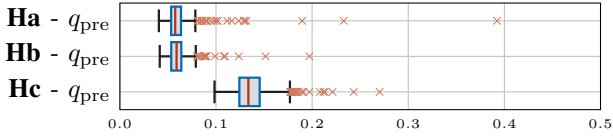
\begin{figure}[t]
	\centering
	\begin{tikzpicture}
		\begin{axis}[
			width=0.9\linewidth,
			height=3cm,
			xmin=0,
			xmax=0.5,
			ymin=0.5,
			ymax=3.5,
			boxplot/draw direction=x,
			xmajorgrids,
			ymajorgrids,
			xtick style = {draw=none},
			xticklabel style={
				font=\tiny,
				/pgf/number format/fixed,
				/pgf/number format/zerofill,
				/pgf/number format/precision=1,
			},
			scaled x ticks=false,
			ytick style = {draw=none},
			ytick={1,...,3},
			yticklabels={%
				\ref{en:sim:H:δiss} - $q_{\mathrm{pre}}$,
				\ref{en:sim:H:δbibs} - $q_{\mathrm{pre}}$,
				\ref{en:sim:H:sta} - $q_{\mathrm{pre}}$,
			},
		]
			\addplot+ [my boxplot] table [y=diss] {latex_data/bp_hodgkin_pre};
			\addplot+ [my boxplot] table [y=dbibs] {latex_data/bp_hodgkin_pre};
			\addplot+ [my boxplot] table [y=sta] {latex_data/bp_hodgkin_pre};

		\end{axis}
	\end{tikzpicture}
	\caption{Box plots of the performance indexes in prediction of system $\mathsf{H}$.}
	\label{fig:perf_hodgkin_pre}
\end{figure}
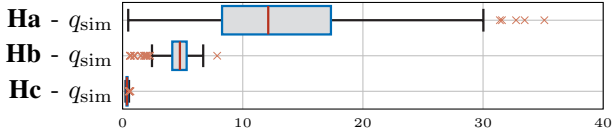
\begin{figure}[t]
	\centering
	\begin{tikzpicture}
		\begin{axis}[
			width=0.9\linewidth,
			height=3cm,
			xmin=0,
			xmax=40,
			ymin=0.5,
			ymax=3.5,
			boxplot/draw direction=x,
			xmajorgrids,
			ymajorgrids,
			xtick style = {draw=none},
			xticklabel style={
				font=\tiny,
				/pgf/number format/fixed,
				/pgf/number format/zerofill,
				/pgf/number format/precision=0,
			},
			scaled x ticks=false,
			ytick style = {draw=none},
			ytick={1,...,3},
			yticklabels={%
				\ref{en:sim:H:δiss} - $q_{\mathrm{sim}}$,
				\ref{en:sim:H:δbibs} - $q_{\mathrm{sim}}$,
				\ref{en:sim:H:sta} - $q_{\mathrm{sim}}$,
			},
		]
			\addplot+ [my boxplot] table [y=diss] {latex_data/bp_hodgkin_sim};
			\addplot+ [my boxplot] table [y=dbibs] {latex_data/bp_hodgkin_sim};
			\addplot+ [my boxplot] table [y=sta] {latex_data/bp_hodgkin_sim};

		\end{axis}
	\end{tikzpicture}
	\caption{Box plots of the performance indexes in simulation of system $\mathsf{H}$.}
	\label{fig:perf_hodgkin_sim}
\end{figure}
\begin{figure}[t]
	\centering
	\begin{tikzpicture}
		\begin{axis}[
			name=pre,
			width=\linewidth,
			height=3cm,
			xmin=0,
			xmax=50,
			xmajorgrids,
			ymajorgrids,
			xticklabels=none,
			y tick label style={font=\tiny,},
		]

		\addplot+[mark=none, solid, smooth, semithick, color=Black] table [x=t, y=true] {latex_data/ts_hodgkin_y};
		\addplot+[mark=none, solid, smooth, semithick, color=BrickRed] table [x=t, y=sta_pre] {latex_data/ts_hodgkin_y};
		\addplot+[mark=none, solid, smooth, semithick, color=NavyBlue] table [x=t, y=diss_pre] {latex_data/ts_hodgkin_y};

		\end{axis}
		\begin{axis}[
			name=sim,
			at=(pre.south),
			anchor=north,
			yshift=-1pt,
			width=\linewidth,
			height=3cm,
			xmin=0,
			xmax=50,
			xmajorgrids,
			ymajorgrids,
			xlabel = {$t\,[\mathrm{s}]$},
			x label style={font=\scriptsize,},
			x tick label style={font=\tiny,},
			y tick label style={font=\tiny,},
		]

		\addplot+[mark=none, solid, smooth, semithick, color=Black] table [x=t, y=true] {latex_data/ts_hodgkin_y};
		\addplot+[mark=none, solid, smooth, semithick, color=BrickRed] table [x=t, y=sta_sim] {latex_data/ts_hodgkin_y};
		\addplot+[mark=none, solid, smooth, semithick, color=NavyBlue] table [x=t, y=diss_sim] {latex_data/ts_hodgkin_y};

		\end{axis}
	\end{tikzpicture}
	\caption{
		The top (bottom) plot shows the comparison between the prediction (simulation) sequence of~\ref{en:sim:H:sta} (solid red line), the prediction (simulation) sequence of~\ref{en:sim:H:δiss} (solid blue line), and the output of System $\mathsf{H}$ (solid black line).
		All the sequences are computed by taking the estimated model from the Monte Carlo run with median performance index using a newly generated validation dataset.
	}
	\label{fig:ts_hodgkin}
\end{figure}

\section{Concluding Remarks}
\label{sec:end}

Identifying nonlinear models possessing certain desired stability properties is of crucial importance in applications.
Yet, existing results are limited in scope, and a systematic procedure to learn stable nonlinear systems is still a challenging problem.
In this article, we make a step towards the solution of this problem showing how stability properties can be effectively included in a kernel-based learning procedure.
In particular, we embed these properties in the hyperparameters' selection algorithm, and we guarantee stability by constraining the usual Tikhonov regularization optimization problem.

The proposed approach has the merit of offering a systematic non-parametric procedure to guarantee stability of the learned models.
It is also a flexible solution allowing different kernel functions to be used.
Moreover, the stability guarantees that can be enforced refer to both the input to the system and to the prediction error.
In turn, stability with respect to the prediction error is shown to have considerable benefits when the model is used for simulation.
As a drawback, however, the stability conditions enforced (formally given in Propositions~\ref{prop:stability:f} and~\ref{prop:δstability:f}) may be quite conservative in general, and their restrictiveness increases with the model dimension.
In practice, this issue can have a detrimental impact on the data-fitting performances in certain cases.
Future research efforts will focus on quantifying this detriment and on addressing with this problem, which may require continuity conditions weaker than Lipschitz-continuity, and non-uniform contraction of the predictor.

\appendix

\subsection{Proof of Theorem~\ref{thm:representer:stable_c}}
\label{thm:representer:stable_c::proof}

For every $c\in\R^{n-m}$, define $\phi_c\coloneq \sum_{i=m+1}^n c_i k_\eta^{\bar z_i} \in \mathcal{H}(k_\eta)$ as in~\eqref{eq:def:f_c}, and let $\mathcal{H}_0 = \{ \phi_c \st c \in \R^{n-m} \} \subseteq \mathcal{H}(k_\eta)$ and $\mathcal{H}_\bot = \{ g \in \mathcal{H}(k_\eta) \st \forall f \in \mathcal{H}_0, \inner*{f}{g}_{k_\eta} = 0 \}$.
Since $\mathcal{H}_0$ is finite-dimensional, and therefore closed, for each $f \in \mathcal{H}(k_\eta)$ there exists $c \in \R^{n-m}$ and $g \in \mathcal{H}_\bot$ such that $f = \phi_c + g$.
Also, since $\mathcal{H}_0$ and $\mathcal{H}_\bot$ are orthogonal, $\norm{f}_{k_\eta}^2 = \norm{\phi_c}_{k_\eta}^2 + \norm{g}_{k_\eta}^2$.
Furthermore, from the definition of $\mathcal{H}_\bot$, $g(\bar{z}_i) = 0$ for all $i \in \{m+1, \ldots, n \}$.
From these facts, we obtain that $f \in \mathcal{H}(k_\eta)$ solves~\eqref{eq:opt_c:stable} if and only if there exists $(c, g) \in \R^{n-m} \times \mathcal{H}_\bot$ such that $f = \phi_c + g$, and $(c, g)$ solves
\begin{equation} \label{eq:opt_c:stable:divided}
	\min_{c \in \R^{n-m}, g \in \mathcal{H}_\bot}
	B(c, g)
	\hspace{3ex}
	\textnormal{ s.t. }
	\hspace{3ex}
	C(c, g) \le 0,
\end{equation}
where
\begin{align*}
	B(c, g)
	&
	=
	\sum_{i=m+1}^n \big( \phi_c(\bar{z}_i) - \bar{y}_i \big)^2 
	+ \beta \norm{\phi_c}^2_{k_\eta}
	+ \beta \norm{g}^2_{k_\eta}
	\\
	C(c, g)
	&
	=
	m \norm{\phi_c}^2_{k_\eta}
	+ m \norm{g}^2_{k_\eta}
	- \chi.
\end{align*}
Since, for all $g\in\mathcal{H}_\bot$, $B(c, 0) \le B(c, g)$ and $C(c, 0) \le C(c, g)$, we then conclude that, if there exists a solution of~\eqref{eq:opt_c:stable:divided} of the form $(c, 0)$, then $f=\phi_c$ solves~\eqref{eq:opt_c:stable}.
The reminder of the proof shows that $c$ given in~\eqref{eq:vec_c:stable} is such that this is the case, thereby concluding the proof.

First, notice that $(c, 0)$ is a solution of~\eqref{eq:opt_c:stable:divided} if and only if $c$ solves
\begin{equation} \label{eq:opt_c:stable:divided-bar}
	\min_{c \in \R^{n-m}}
	\bar{B}(c)
	\hspace{3ex}
	\textnormal{ s.t. }
	\hspace{3ex}
	\bar{C}(c) \le 0,
\end{equation}
where $\bar{B}(c) = B(c, 0)$ and $\bar{C}(c) = C(c, 0)$.
By definition of $K_\eta$, and since $k_\eta(a,b) = \inner{k_\eta^{a}}{k_\eta^{b}}_{k_\eta}$ \cite[Lem.~4.19]{Steinwart2008a}, for every $c\in\R^{n-m}$ we have
\begin{equation} \label{eq:opt_c:stable:norm}
	\norm{\phi_c}_{k_\eta}^2
	=
	\sum_{i=m+1}^n \sum_{j=m+1}^n c_i c_j \inner*{k_\eta^{\bar{z}_i}}{k_\eta^{\bar{z}_j}}_{k_\eta}
	=
	c^\top K_\eta c.
\end{equation}
In view of~\eqref{eq:opt_c:stable:norm}, we thus obtain $\bar{B}(c) = c^\top ( K_\eta + \beta I_n ) K_\eta c - 2 \bar{y}^\top K_\eta c + \bar{y}^\top \bar{y}$ and $\bar{C}(c) = m c^\top K_\eta c - \chi$.
Additionally, $\nabla \bar{B}(c) = 2c^\top ( K_\eta + \beta I_n ) K_\eta - 2 \bar{y}^\top K_\eta$ and $\nabla \bar{C}(c) = 2 m c^\top K_\eta$.
As a consequence, $\bar{B}$ and $\bar{C}$ are smooth, convex functions and, hence, the \emph{KKT conditions} are sufficient for optimality~\cite[Sec.~5.5.3]{Boyd2019a}.
Namely, if $(c,\lambda)\in\R^{n-m}\times [0,\infty)$ satisfies
\begin{equation} \label{eq:opt_c:stable:KKT}
	\nabla\bar{J}(c) + \lambda \nabla\bar{C}(c) = 0_n^\top,
	\hspace{2ex}
	\lambda \bar{C}(c) = 0,
	\hspace{2ex}
	\bar{C}(c) \le 0,
\end{equation}
then $c$ is an optimal solution of~\eqref{eq:opt_c:stable:divided-bar}.
The first of~\eqref{eq:opt_c:stable:KKT} is satisfied by $c = ( K_\eta + (\beta + \lambda m) I_{n-m} )^{-1} \bar{y}$.
Hence, if there exists $\alpha \in [\beta, \infty)$ such that
\begin{equation*}
	\mathsf{C2} \,:\, \frac{\alpha - \beta}{m} \gamma(\alpha) = 0,
	\hspace{3ex}
	\mathsf{C3} \,:\,\gamma(\alpha) \le 0,
\end{equation*}
where $\gamma$ is defined in~\eqref{eq:def:gamma}, then~\eqref{eq:opt_c:stable:KKT} are satisfied by $c = ( K_\eta + \alpha I_{n-m} )^{-1} \bar{y}$ and $\lambda = \frac{\alpha - \beta}{m}$.

We now find $\alpha$ satisfying $\mathsf{C2}$ and $\mathsf{C3}$, thereby completing the proof.
First, notice that $\gamma$ is non-increasing and $\lim_{\alpha\to\infty} \gamma(\alpha) = - \chi$.
Hence, if $\gamma$ is not always negative, there exists a $\bar{\alpha} \in [0,\infty)$ such that $\gamma(\bar{\alpha}) = 0$.
Otherwise, we let $\bar{\alpha} = 0$.
If $\beta \le \bar{\alpha}$, then $\mathsf{C2}$, $\mathsf{C3}$, and $\alpha \ge \beta$ are all satisfied by the choice $\alpha = \bar{\alpha}$ since $\gamma(\bar{\alpha}) \le 0$ by definition of $\bar{\alpha}$.
Otherwise, if $\beta > \bar{\alpha}$, we notice that $\alpha = \beta$ trivially fulfills $\mathsf{C2}$.
Moreover, $\mathsf{C3}$ is also satisfied because $\gamma$ is a non-increasing function and $\gamma(\bar{\alpha}) \le 0$.
Hence, we conclude the proof by noticing that, in all cases, $\alpha = \max(\bar{\alpha}, \beta)$.

\subsection{Proof of Proposition~\ref{prop:stability:kernel}}
\label{prop:stability:kernel::proof}

Pick $f\in\mathcal{H}(k_\eta)$.
Using the Cauchy-Schwarz inequality, for every $a\in\R^{2m+1}$, we obtain
\begin{equation*}
	\norm*{f(a)}^2
	=
	\norm*{\inner{f}{k_\eta^a}_{k_\eta}}^2
	=
	\inner[\big]{f}{k_\eta^a}_{k_\eta}^2
	\le
	\norm*{f}_{k_\eta}^2 \inner*{k_\eta^a}{k_\eta^a}_{k_\eta}.
\end{equation*}
Since $\inner{k_\eta^a}{k_\eta^a}_{k_\eta}=k_\eta(a,a)$, then $\norm*{f(a)}^2 \le \norm*{f}_{k_\eta}^2 k_\eta(a, a).$
Therefore, Conditions~\eqref{eq:stability:kernel:cond} lead to the implications
\begin{align*}
	\norm*{a}^2 \ge \nu
	&
	\Rightarrow \norm*{f(a)}^2 \le \norm*{f}_{k_\eta}^2 \norm*{a}^2,
	\\
	\norm*{a}^2 < \nu
	&
	\Rightarrow \norm*{f(a)} \le \norm*{f}_{k_\eta} \sqrt{k_\eta(a, a)} \le \norm*{f}_{k_\eta} \sqrt{s},
\end{align*}
which coincide with~\eqref{eq:stability:cond} for $a = \col(x,u)$ ($x \in \R^{2m}$ and $u \in \R$), $\mu = m \norm*{f}_{k_\eta}^2$, $\rho = \nu$, $d = \norm*{f}_{k_\eta} \sqrt{s}$ and $\omega(b) = \norm*{f}_{k_\eta}^2 b^2$.

\subsection{Proof of Proposition~\ref{prop:δstability:f}}
\label{prop:δstability:f::proof}

We divide the proof in several subsections for clarity of exposition.

\subsubsection{$\delta$ISS Lyapunov Candidate}

With $R\coloneq \diag(1,\ldots,m)$, let $P \coloneq \diag(R,R)$ and define the function $V:\R^{2m} \to [0,\infty]$ as $V(x)\coloneq x^\top P x$. Let $A$, $G$ and $H$ be the matrices defined in Section~\ref{sec:problem}.
Then, $P$ satisfies
\begin{equation} \label{eq:identities_P}
	\begin{gathered}
		A^\top P A - P = -I_{2m},
		\hspace{3ex}
		G^\top P G = H^\top P H = m, \\
		G^\top P H = 0,
		\hspace{3ex}
		A^\top P G = A^\top P H = 0_{2m},
	\end{gathered}
\end{equation}
and
\begin{equation} \label{eq:sandwitch_V}
	\forall x\in\R^{2m},
	\hspace{3ex}
	\norm*{x}^2 \le V(x) \le m \norm*{x}^2.
\end{equation}
Fix $x^a, x^b \in \R^{2m}$, $v^a, v^b, \varepsilon^a, \varepsilon^b \in \R$, and let
\begin{equation} \label{eq:def:tilde}
	\begin{aligned}
		\tilde{x} & \coloneq x^a - x^b,
		&
		\tilde{v} & \coloneq v^a - v^b,
		\\
		\tilde{\varepsilon} & \coloneq diff{\varepsilon}^a - diff{\varepsilon}^b,
		&
		\tilde{f} & \coloneq f(x^a,v^a) - f(x^b,v^b).
	\end{aligned}
\end{equation}
By using the Young's inequality $2\tilde{f}\tilde{\varepsilon} \le (1-\mu)\tilde{f}^2 + (1-\mu)^{-1}\tilde{\varepsilon}^2$~\cite{Young1912a} and~\eqref{eq:identities_P} we obtain
\begin{equation} \label{eq:DV}
	\begin{aligned}
		&V\big( A \tilde{x} + G \tilde{f} + G \tilde{\varepsilon} + H \tilde{v} \big) - V( \tilde{x} ) 
		=
		\\
		& \quad =
		-| \tilde{x} |^2 + m\tilde{f}^2 + 2m\tilde{f}\tilde{\varepsilon} + m\tilde{\varepsilon}^2 + m\tilde{v}^2
		\\
		& \quad \le
		-| \tilde{x} |^2 + m(2-\mu)\tilde{f}^{2} + m \frac{2-\mu}{1-\mu} \tilde{\varepsilon}^{2} + m \tilde{v}^{2} .
	\end{aligned}
\end{equation}
Using~\eqref{eq:δstability:cond:contractive} and~\eqref{eq:sandwitch_V}, we thus obtain
\begin{multline} \label{eq:DV_ge_rho}
	\norm{\tilde{x}}^2 + \norm{\tilde{v}}^2 \ge \rho
	\Rightarrow 
	V\big( A \tilde{x} + G \tilde{f} + G \tilde{\varepsilon} + H \tilde{v} \big)
	\\
	\le \lambda V(\tilde{x}) + \varrho_{\mathrm{u}}(\norm*{\tilde{v}}) + \varrho_{\mathrm{e}}(\norm*{\tilde{\varepsilon}}) ,
\end{multline}
with $\lambda \coloneq 1 - \frac{(1-\mu)^2}{m} \in [0,1)$, $\varrho_{\mathrm{u}}(s) \coloneq m ( 2 - \mu ) \omega(s) + m s^{2}$, and $\varrho_{\mathrm{e}}(s) \coloneq m \frac{2-\mu}{1-\mu} s^{2}$.
Notice that $\varrho_{\mathrm{u}}$ and $\varrho_{\mathrm{e}}$ are of class-$\mathcal{K}$.

\subsubsection{Proof of~\ref{prop:δstability:f:BIBS}}
\label{prop:δstability:f::proof:BIBS}

Fix $q \in [0,\infty)$ and pick arbitrarily two solutions $(x^a,\tilde{v}^a,\tilde{\varepsilon}^a)$ and $(x^b,\tilde{v}^b,\tilde{\varepsilon}^b)$ of $\Sigma(f)$ such that $\norm*{x^a_0-x^b_0} \le q$, $\norm*{\tilde{v}^a-\tilde{v}^b}_\infty \le q$, and $\norm*{\tilde{\varepsilon}^a - \tilde{\varepsilon}^b}_\infty \le q$.
For simplicity, we use the error variables~\eqref{eq:def:tilde}.
Moreover, define
\begin{align*} %
	c_1 & \coloneq \sqrt{\frac{2(\varrho_{\mathrm{u}}(q) + \varrho_{\mathrm{e}}(q))}{(1-\lambda)m}},
	\hspace{2ex}
	c_2 \coloneq \max \left\{ q , c_1, \sqrt{\rho} \right\}, \nonumber
	\\
	c_3 & \coloneq m c_2^2,
	\hspace{2ex}
	c_4 \coloneq \sup_{\norm*{\tilde{x}} \le c_2, \, \norm*{\tilde{v}} \le q,\, \norm*{\tilde{\varepsilon}} \le q} \norm*{\tilde{f}^2 + 2\tilde{f}\tilde{\varepsilon}} + 2q^2, \nonumber
	\\
	c_5 & \coloneq c_3 + 2m c_4.
\end{align*}
As $\lambda\in[0,1)$, then $c_1>0$.
Moreover, $c_4$ exists since, in view of~\eqref{eq:δstability:cond}, $\tilde f$ is bounded whenever so are $\tilde{x}$ and $\tilde{v}$.
Also, in view of the first equality in~\eqref{eq:DV}, we have
\begin{equation} \label{eq:V_le_c3_implies_vplus_le_c5}
	\forall t\in\N,\ V(\tilde{x}_t)\le c_3 \Rightarrow V(\tilde{x}_{t+1}) \le c_3 + m c_4 \le c_5.
\end{equation}
The proof is carried out by contradiction.
Suppose that there exists no $w\in[0,\infty)$ such that $\norm*{\tilde{x}}_\infty \le w$.
Then, in view of~\eqref{eq:sandwitch_V}, there exists $\bar{t}$ such that $V(\tilde{x}_{\bar{t}}) \ge \norm*{\tilde{x}_{\bar{t}}}^2 > c_5$ (otherwise $\norm*{\tilde{x}}_\infty\le w$ with $w=\sqrt{c_5}$).
In view of~\eqref{eq:sandwitch_V}, for every $t \in \N$, $|\tilde{x}_t| \le c_2$ implies $V(\tilde{x}_t) \le m |\tilde{x}_t|^2 \le m c_2^2 = c_3$ and thus, in view of~\eqref{eq:V_le_c3_implies_vplus_le_c5}, $V(\tilde{x}_{t+1}) \le c_5$.
Therefore, since $|\tilde{x}_0| \le q \le c_2$, we have $V(\tilde{x}_0) \le c_3$ and $V(\tilde{x}_1) \le c_5$.
Thus, since $V(\tilde{x}_{\bar{t}}) > c_5$, there exist $t_0, t_1 \in \N$ satisfying $1 \le t_0 < t_1 \le \bar{t}$ and such that $V(\tilde{x}_t) \in [c_3, c_5]$ for every $t \in \{t_0, \ldots, t_1-1\}$ and $V(\tilde{x}_{t_1}) > c_5$.
In view of~\eqref{eq:V_le_c3_implies_vplus_le_c5}, $V(\tilde{x}_{t_1 - 1}) \ge c_3$. Hence,
\begin{equation*}
	|\tilde{x}_{t_1 - 1}|^2 + |\tilde{v}_{t_1 - 1}|^2
	\ge |\tilde{x}_{t_1 - 1}|^2
	\ge \frac{V(\tilde{x}_{t_1-1})}{m}
	\ge \frac{c_3}{m}
	= c_2^2
	\ge \rho.
\end{equation*}
Therefore, from~\eqref{eq:DV_ge_rho}, we obtain $V(\tilde{x}_{t_1}) \le \lambda V(\tilde{x}_{t_1 - 1}) + \varrho_{\mathrm{u}}(q) + \varrho_{\mathrm{e}}(q)$.
Since $c_3 \le V(\tilde{x}_{t_1 - 1}) \le c_5$ and $c_3 = m c_2^2 \ge m c_1^2$, we can write
\begin{align*}
	V(\tilde{x}_{t_1})
	& \le V(\tilde{x}_{t_1 - 1}) - (1-\lambda) V(\tilde{x}_{t_1 - 1}) + \varrho_{\mathrm{u}}(q) + \varrho_{\mathrm{e}}(q)
	\\
	& \le c_5 - (1-\lambda) c_3 + \varrho_{\mathrm{u}}(q) + \varrho_{\mathrm{e}}(q)
	\\
	&
	\le c_5 - m(1-\lambda) c_1^2 + m \frac{1-\lambda}{2} c_1^2
	\le c_5 - \frac{m(1-\lambda)}{2} c_1^2.
\end{align*}
Since $c_1 \ge 0$, we obtain $V(\tilde{x}_{t_1}) \le c_5$, a contradiction.

\subsubsection{Preliminaries to the proof of~\ref{prop:δstability:f:wAG}}
\label{prop:δstability:f::proof:unif-cont}

Before proving~\ref{prop:δstability:f:wAG}, we state and prove two technical lemmas dealing with uniformly continuous functions.
\begin{lemma} \label{lem:f_unif_continuous_incr_bounded}
	Let $p \in \N$, $r \in [0, \infty)$, and $\varphi : \R^p \to \R$ be a uniformly continuous function.
	Then $(x,y) \mapsto \varphi(x)-\varphi(y)$ is bounded on $\{(x,y) \in \R^p \times \R^p \st |x-y| \le r \}$.
\end{lemma}
\begin{proof}
	Since $\varphi$ is uniformly continuous, there exists $\delta > 0$ such that, for every $a,b\in\R^p$, $|a-b| \le \delta \Rightarrow |\varphi(a)-\varphi(b)| \le 1$.
	Let $N\in\N_{>0}$ be such that $N\ge \frac{r}{\delta}$.
	Pick $x,y \in \R^p$ satisfying $|x-y| \le r$.
	For each $k \in \{ 0, \ldots, N \}$, let $z_k \coloneq \frac{N-k}{N} x + \frac{k}{N} y$.
	Then, $| z_{k+1} - z_{k} | = \frac{1}{N} |x-y| \le \frac{r}{N} \le \delta$.
	Since $x = z_0$ and $y = z_N$, we then obtain
	$| \varphi(x)-\varphi(y) |
		= \big| \varphi(z_0) \pm \textstyle\sum_{k=1}^{N-1} \varphi(z_k) - \varphi(z_N) \big|
 		\le \sum_{k=0}^{N-1}\left| \varphi(z_k) - \varphi(z_{k+1}) \right|
		\le N$.
\end{proof}
\begin{lemma} \label{lem:continuity_sup}
	Let $p \in \N$, $\varphi : \R^p \to \R$ be a uniformly continuous function, and $g : \R \to [0, \infty)$ be a continuous function such that $g(0)=0$.
	Then, the function%
	\footnote{Notice that $k$ is defined on $[0,\infty)$ in view of Lemma~\ref{lem:f_unif_continuous_incr_bounded}.}
	$k : [0, \infty) \to [0, \infty)$ defined as
	$k(s) \coloneq \sup_{|x-y| \le s} g( \varphi(x) - \varphi(y) )$
	is continuous, non-decreasing, and $k(0) = 0$.
\end{lemma}
\begin{proof}
	The only non-trivial claim to prove is continuity of $k$.
	In the following, for $h \in [0,\infty)$, we let $\Omega_h \coloneq \{ (x,y)\in\R^p\times \R^p \st |x-y| \le h\}$.
	Pick $s \in [0, \infty)$ and $\iota \in (0, \infty)$, and let $M_{s+1} \coloneq \sup_{(x,y)\in\Omega_{s+1} } | \varphi(x)-\varphi(y) | \ge 0$, which exists in view of Lemma~\ref{lem:f_unif_continuous_incr_bounded}.
	Since $g$ is continuous, it is uniformly continuous on $[-M_{s+1}, M_{s+1}]$.
	Hence, there exists $\delta^g(\iota) \in (0, \infty)$ such that, for every $a,b\in[-M_{s+1},M_{s+1}]$, $|a-b|\le\delta^g(\iota) \Rightarrow |g(a)-g(b)| \le \frac{\iota}{2}$.
	Since $\varphi$ is uniformly continuous, there exists $\delta^\varphi(\iota) \in (0, \infty)$ such that, for every $(x,y),(x',y')\in \Omega_{s+1}$, $|x-x'|+|y-y'| \le \delta^\varphi(\iota)\Rightarrow| (\varphi(x)-\varphi(y)) - (\varphi(x')-\varphi(y')) | \le | \varphi(x)-\varphi(x') | + | \varphi(y)-\varphi(y') | \le \delta^g(\iota)$, which in turn implies $| g(\varphi(x)-\varphi(y)) - g(\varphi(x')-\varphi(y'))| \le \frac{\iota}{2}$.

	With $\delta(\iota) \coloneq \min\{1,\delta^\varphi(\iota)\}$, let $(\bar{x},\bar{y}) \in \Omega_{s+\delta(\iota)}$ be such that $k(s+\delta(\iota)) \le g(\varphi(\bar x)-\varphi(\bar y)) + \frac{\iota}{2}$ (which exists by definition of $k$), and define
	$\bar{x}' \coloneq \bar{x} -\frac{\delta(\iota)}{2(s+\delta(\iota))} (\bar{x}-\bar{y})$ and $\bar{y}' \coloneq \bar{y} +\frac{\delta(\iota)}{2(s+\delta(\iota))} (\bar{x}-\bar{y})$.
	Then, $|\bar{x}'-\bar{y}'| = \frac{s}{s+\delta(\iota)} |\bar{x}-\bar{y}|\le s$.
	Hence, $(\bar{x}',\bar{y}')\in\Omega_{s}$.
	Moreover, $|\bar{x}-\bar{x}'|+|\bar{y}-\bar{y}'| = \frac{\delta(\iota)}{s+\delta(\iota)} |\bar{x}-\bar{y}| \le \delta(\iota)$.
	As $\delta(\iota)\le 1$ and $\delta(\iota)\le \delta^\varphi(\iota)$, then
	\begin{gather*}
		| \varphi(\bar{x}) -\varphi(\bar{y}) | \le M_{s+1},
		\hspace{1ex}
		| \varphi(\bar{x}')-\varphi(\bar{y}') | \le M_{s+1},
		\\
		| g( \varphi(\bar{x})-\varphi(\bar{y}) ) - g( \varphi(\bar{x}')-\varphi(\bar{y}') ) | \le \frac{\iota}{2} .
	\end{gather*}
	Therefore, since $k$ is non-decreasing, and by definition of $(\bar{x},\bar{y})$, we obtain
	\begin{align*}
		&
		| k(s+\delta(\iota)) - k(s) | = k(s+\delta(\iota)) - k(s) \\
		& \hspace{3ex}\le k(s+\delta(\iota)) - g( \varphi(\bar{x}')-\varphi(\bar{y}') )
		\\
		&
		\hspace{3ex} \le g( \varphi(\bar x)-\varphi(\bar y) ) + \frac{\iota}{2} - g( \varphi(\bar{x}')-\varphi(\bar{y}'))
		\\ &
		\hspace{3ex} \le | g( \varphi(\bar x)-\varphi(\bar y) ) - g( \varphi(\bar x')-\varphi(\bar y') ) | + \frac{\iota}{2}
		\le \iota .
	\end{align*}
	The proof then follows by the arbitrariness of $s$ and $\iota$.
\end{proof}

\subsubsection{Proof of~\ref{prop:δstability:f:wAG}}
Assume that $f$ is uniformly continuous.
Pick arbitrarily two solutions $(x^a,\tilde{v}^a,\tilde{\varepsilon}^a)$ and $(x^b,\tilde{v}^b,\tilde{\varepsilon}^b)$ of $\Sigma(f)$.
For simplicity, we use the error variables~\eqref{eq:def:tilde}.
Let $q'\coloneq \max\{ |\tilde{v}|_\infty, |\tilde{\varepsilon}|_\infty\}$ and define the quantities
\begin{multline*}
	c_1' \coloneq \sqrt{\frac{2(\varrho_{\mathrm{u}}(q') + \varrho_{\mathrm{e}}(q'))}{(1-\lambda)m}},
	\quad
	c_2' \coloneq \max \left\{ q' , c_1' ,\sqrt{\rho} \right\},
	\\
	c_3' \coloneq m (c_2')^2,
	\quad
	c_4' \coloneq \sup_{|\tilde{x}| \le c_2', \, |\tilde{v}| \le q',\, |\tilde{\varepsilon}| \le q'} \left| \tilde{f}^2 + 2\tilde{f}\tilde{\varepsilon}\right| + 2(q')^2.
\end{multline*}
As in Section~\ref{prop:δstability:f::proof:BIBS},~\eqref{eq:δstability:cond} guarantee that $c_4' \ge 0$ exists.
Moreover, proceeding as in Section~\ref{prop:δstability:f::proof:BIBS}, we find that, for all $t\in\N$, $V(\tilde x_{t}) \ge c_3'$ implies $|\tilde{x}_t| \ge (c_2')^2 \ge \rho$ and $V(\tilde{x}_{t+1}) \le V(\tilde x_t) - \frac{(1-\lambda) m}{2} (c_1')^2$.
Furthermore, the first equality of~\eqref{eq:DV} yields $V(\tilde x_{t}) \le c_3' \Rightarrow V(\tilde x_{t+1}) \le c_3' + m c_4'$.
Assuming without loss of generality $q'>0$, these two inequalities and~\eqref{eq:sandwitch_V} suffice to claim the existence of a $\bar{t} \in \N$ such that
\begin{align} \label{eq:proof-stability:beta}
	& \forall t \ge \bar t, & & | \tilde{x}_t | \le \sqrt{c_3'+m c_4'} \eqqcolon \beta(q',\rho) ,
\end{align}
where we highlighted the dependency of $\beta$ from $q'$ and $\rho$.

Next, we can write
\begin{align*}
	\sup_{
		\substack{
			|\tilde{x}| \le \sqrt{c_3'},\\
			|\tilde{v}| \le q'
		}} \tilde{f}^2
	\le \sup_{| (\tilde x,\tilde u) | \le 2\max\{ \sqrt{c_3'}, q'\}} \tilde f^2
	\eqqcolon \omega_1(\max\{\sqrt{c_3'},q'\}).
\end{align*}
Lemma~\ref{lem:continuity_sup} applied with $p=2m+1$, $\varphi=f$, and $g(s)=s^2$ implies that $\omega_1$ is continuous, non-decreasing, and $\omega_1(0)=0$.
We also have $\omega_1(\max\{\sqrt{c_3'},q'\}) \le \omega_1(\sqrt{c_3'}) + \omega_1(q')$.
As a consequence, we can bound $c_4'$ as
\begin{align*}
	c_4'
	\le 2 \cdot \sup_{ | \tilde{x} | \le \sqrt{c_3'},\,| \tilde{v} | \le q'} \tilde f^2 + 3 (q')^2
	= \frac{\omega_2(c_3') - c_3' + \omega_3(q')}{m}
\end{align*}
in which $\omega_2(c_3') \coloneq c_3'+2m\omega_1(\sqrt{c_3'})$ and $\omega_3(q')\coloneq 2m\omega_1(q')+3m(q')^2$ are both of class-$\mathcal{K}$ since they are the sum of a class-$\mathcal{K}$ function (the map $c_3'\mapsto c_3'$ and $q'\mapsto 3m(q')^2$, respectively) and the functions $2m\omega_1(\sqrt{\cdot})$ and $2m\omega_1$, with $\omega_1$ that has the previously-claimed properties.
As a consequence, we have $\beta(q',\rho)^2 \le \omega_2(c_3') + \omega_3(q')$.
Since $c_3' \le 2m((\max\{c_1',q'\})^2+\rho)$, and $c_1'=c_1'(q')$, seen as a function of $q'$, is of class-$\mathcal{K}$, then by using the triangle inequality $\omega_2(s+z)\le \omega_2(2s)+\omega_2(2z)$~\cite[Eq.~(12)]{sontag_smooth_1989}, we can finally write $\beta(q',\rho) \le b(\rho) + \kappa(q')$, where $\kappa(q') \coloneq \sqrt{ \omega_2(4m \max\{c_1'(q'),q'\}^2) + \omega_3(q') }$ and $b(\rho) \coloneq \sqrt{\omega_2(4m \rho)}$.
Using the definition of $q'$, from~\eqref{eq:proof-stability:beta} we get
\begin{equation*}
	\limsup_{t\to\infty} | \tilde{x}_t | \le b(\rho) + \kappa \left( \max\{|\tilde{v}|_\infty ,|\tilde{\varepsilon}|_\infty\} \right),
\end{equation*}
and the result follows by the same arguments of~\cite[Lem.~II.1]{Sontag1996a} and since $b$ is of class-$\mathcal{K}$.

\subsubsection{Proof of~\ref{prop:δstability:f:ISS}}

If $\rho=0$, then $\delta$ISS directly follows from~\eqref{eq:DV_ge_rho}, because the function $(x,y) \mapsto V(x-y)$ is a $\delta$ISS Lyapunov function.
See,~\cite[Thm.~8]{tran_incremental_2016} and~\cite{Scandella2023b} for more details.

\subsection{Proof of Proposition~\ref{prop:δstability:kernel}}
\label{prop:δstability:kernel::proof}

Pick $f\in\mathcal{H}(k_\eta)$.
For every $a,b\in\R^{2m+1}$, $| f(a) - f(b) |^2
	=
	| \inner{f}{k_\eta^a}_{k_\eta} - \inner{f}{k_\eta^{b}}_{k_\eta} |^2
	=
	\inner{f}{k_\eta^a - k_\eta^{b}}_{k_\eta}^2
	\le
	| f |_{k_\eta}^2 \inner*{k_\eta^a - k_\eta^{b}}{k_\eta^a - k_\eta^{b}}_{k_\eta}$,
where we used the Cauchy-Schwarz inequality.
Since, for every $a,b \in \mathbb{R}^{2m+1}$, $k_\eta(a,b) = \inner{k_\eta^a}{k_\eta^b}_{k_\eta}$, we obtain $| f(a) - f(b) |^2 \le | f |_{k_\eta}^2 h_\eta(a, b)$, in which, we recall, $h_\eta(a,b)\coloneq k_\eta(a,a)-2k_\eta(a,b)+k_\eta(b,b)$.
Then, in view of~\eqref{eq:δstability:kernel:cond}, the proof follows by the same arguments of that of Proposition~\ref{prop:stability:kernel} (see Appendix~\ref{prop:stability:kernel::proof}).

\subsection{Proof of Proposition~\ref{prop:k:linear}}
\label{prop:k:linear::proof}

Since, for every $\eta=(\tau, \sigma) \in \Phi_k$, $k_\eta$ is continuous, then, by Proposition~\ref{prop:k:cont-uni_cont}, $k$ is $\rho$-viable if Condition~\eqref{eq:stability:kernel:cond:contractive} holds.
Moreover, $k_\eta(a,a) = \tau \norm*{a}^2 + \sigma$.
Hence, Condition~\eqref{eq:stability:kernel:cond:contractive} is equivalent to the existence of $\nu \in [0, \rho] \cap \R$ such that, for every $z \in [\nu,\infty)$, $g(z) \coloneqq (\tau - 1) z + \sigma \le 0$.
If $\tau \in (1, \infty)$, $g$ is increasing and strictly positive in $(0,\infty)$.
Therefore, the condition is not satisfied for any $\rho \in [0, \infty]$.
Instead, if $\tau \in [0,1]$, $g(z)$ is non-increasing.
Therefore, its maximum in the interval $[\nu,\infty)$ is $g(\nu) = (\tau - 1) \nu + \sigma$.
As $g(\nu)\ge g(\rho)$ for all $\nu\le\rho$, then we conclude that $k$ is $\rho$-viable if and only if $\tau \in [0,1]$ and $\sigma \in [0,\rho (1-\tau)]$.

As for what concerns $\rho$-$\delta$viability, notice that, for every $(\tau, \sigma) \in \Phi_k$, $h_\eta(a,b) = \tau \norm*{a - b}^2$. Hence,~\eqref{eq:δstability:kernel:cond:boundedness} trivially holds for every $(\tau, \sigma) \in \Phi_k$ with any $s>\tau\nu$, whereas~\eqref{eq:δstability:kernel:cond:contractive} holds if and only if $\tau\le 1$.

\subsection{Proof of Proposition~\ref{prop:k:polynomial}}
\label{prop:k:polynomial::proof}

We recall that $k_\eta(a,b) = \Gamma_\eta(a)^\top \Gamma_\eta(b)$ for every $a,b \in \R^{2m+1}$ and $\eta \in \Phi_k$, where $\Gamma_\eta$ is the vector-valued function whose entries are all the possible $\eta$th degree ordered products of its argument~\cite[Prop.~2.1]{Scholkopf2018a}.
By definition, $k$ is $\rho$-viable only if~\eqref{eq:stability:kernel:cond:contractive} holds, namely, only if there exists $\nu \in [0, \rho]$ such that, for all $\norm*{a}\ge \nu$, $\norm*{\Gamma_\eta(a)}^2 \le \norm*{a}^2$, which is never satisfied because $\Gamma_\eta$ is a polynomial function with degree $\eta>1$.
Therefore, $\Theta_k^\rho = \varnothing$.
The equality $\Delta_k^\rho = \varnothing$ can be obtained by similar arguments.

\subsection{Proof of Theorem~\ref{thm:k:stationary}}
\label{thm:k:stationary::proof}

First, we focus on $\Theta^\rho_k$.
Since $k_\eta(a,a) = \bar{k}_\eta(0_{2m+1})$, Conditions~\eqref{eq:stability:kernel:cond} read
\begin{align*}
	\forall a \in \R^{2m+1},
	& & \norm*{a}^2 \ge \nu
	& \Rightarrow \bar{k}_\eta(0_{2m+1}) \le \norm*{a}^2,
	\\
	\forall a \in \R^{2m+1},
	& & \norm*{a}^2 < \nu
	& \Rightarrow \bar{k}_\eta(0_{2m+1}) \le s .
\end{align*}
The second condition is fulfilled with any $\nu\in[0, \infty)$ and with $s\coloneq\bar{k}_\eta(0_{2m+1})$.
Instead, the first condition holds if and only if $\bar{k}_\eta(0_{2m+1}) \le \nu$.
Therefore, the claim of the theorem holds since $\nu \in [0, \rho]\cap\R$ such that $\bar{k}_\eta(0_{2m+1}) \le \nu$ exists if and only if $\rho \ge \bar{k}_\eta(0_{2m+1})$.

Regarding $\Delta^\rho_k$, let $g(z) \coloneq 2\bar{k}_\eta(0_{2m+1}) - 2\bar{k}_\eta(z)$ for all $z \in \R^{2m+1}$.
Then, for all $a,b \in \R^{2m+1}$, $h_\eta(a,b) = g(a-b)$.
Therefore, Conditions~\eqref{eq:δstability:kernel:cond} are equivalent to
\begin{subequations} \label{eq:proof:k:stationary}
	\begin{align}
		\forall z \in \R^{2m+1},
		& & \norm*{z}^2 \ge \nu
		& \Rightarrow g(z) \le \norm*{z}^2, \label{eq:proof:k:stationary:c1}
		\\
		\forall z \in \R^{2m+1},
		& & \norm*{z}^2 < \nu
		& \Rightarrow g(z) \le s. \label{eq:proof:k:stationary:c2}
	\end{align}
\end{subequations}
Now, notice that, for every $z \in \R^{2m+1}$,
\begin{multline}
	\norm*{\bar{k}_\eta(z)}
	=
	\norm*{k_\eta(z,0_{2m+1})}
	=
	\norm*{\inner{k^z_\eta}{k^{0_{2m+1}}_\eta}_{k_\eta}}
	\\
	\le
	\sqrt{k_\eta(z,z) k_\eta(0_{2m+1},0_{2m+1})}
	=
	\bar{k}_\eta(0_{2m+1}).
	\label{eq:barkz_le_bark0}
\end{multline}
Therefore, $g(z) \in [0, 4\bar{k}_\eta(0_{2m+1})]$ for all $z \in \R^{2m+1}$.
Hence,~\eqref{eq:proof:k:stationary:c2} always holds with any $\nu\in[0,\infty)$ and with $s = 4\bar{k}(0_{2m+1})$.
Thus, since~\eqref{eq:proof:k:stationary:c1} holds for some $\nu$ only if it holds for all $\nu'\ge \nu$, then $\Delta^\rho_k = \{ \eta \in \Phi_k \st \forall z \in \R^{2m+1}, \norm*{z}^2 \ge \rho \Rightarrow 2\bar{k}_\eta(0_{2m+1}) - 2\bar{k}_\eta(z) \le \norm*{z}^2 \}$.
Finally, the fact that, for every $\rho\in[0,\infty]$, $\Delta^\rho_k \supseteq \{ \eta \in \Phi_k \st 4\bar{k}_\eta(0_{2m+1}) \le \rho \}$ directly follows from~\eqref{eq:proof:k:stationary:c1} and $g(z) \le 4\bar{k}_\eta(0_{2m+1})$.

\subsection{Proof of Proposition~\ref{prop:k:gaussian}}
\label{prop:k:gaussian::proof}

Since $k$ is a stationary kernel structure with $\bar{k}_\eta(0_{2m+1}) = \tau + \sigma$ for every $\eta \in \Phi_k$, the claimed expression of $\Theta_k^\rho$ directly follows from Theorem~\ref{thm:k:stationary}.
In view of Theorem~\ref{thm:k:stationary}, $\Delta^\rho_k$ is the set of $\eta \in \Phi_k$ such that, for every $z \in \R^{2m+1}$, $\norm*{z}^2 \ge \rho$ implies $2\tau - 2\tau \exp ( - \gamma \norm*{z}^2 ) \le \norm*{z}^2$
or, equivalently,
\begin{equation} \label{eq:proof:gaussian:c}
	\forall \zeta \in [\rho, \infty),
	\hspace{0.5ex}
	g(\zeta) \coloneq 2\tau - 2\tau \exp ( - \gamma \zeta ) - \zeta \le 0 .
\end{equation}
Notice that $g$ is smooth, $g(0) = 0$, and its derivative is $g'(\zeta) = 2 \tau \gamma \exp(-\gamma \zeta) - 1$.
First, pick $\rho=0$.
We have $g'(\zeta) \le 0$ for every $\zeta \ge 0$ if and only if $2 \tau \gamma \le 1$ and, in this case, $g(\zeta)\le 0$ for all $\zeta \ge 0$.
Otherwise, $g$ is increasing in $0$.
Hence, $\Delta_k^0 = \{ (\tau,\gamma,\sigma)\in\Phi_k \st 2\tau\gamma \le 1\}$.
Next, pick $\rho\in(0,\infty]$.
If $2\tau\gamma\le 1$, then $(\tau,\gamma,\sigma) \in \Delta_k^0 \subseteq \Delta_k^\rho$.
Instead, if $2 \tau \gamma > 1$, $g$ is increasing in the interval $[0, \gamma^{-1}\log(2 \gamma \tau))$ and decreasing in $[\gamma^{-1}\log(2 \gamma \tau), +\infty)$.
Therefore, $g(\zeta)\le 0$ for all $\zeta\ge \bar\zeta$, in which $\bar\zeta$ is the unique strictly positive solution of $g(\bar\zeta)=0$, which is given by $v(\gamma, \tau)$.
Then,~\eqref{eq:proof:gaussian:c} holds if and only if $\rho\ge v(\gamma, \tau)$, implying $\Delta_k^\rho = \{ (\tau,\gamma,\sigma)\in\Phi_k\st v(\gamma, \tau) \le \rho \}$. 
For $\rho=\infty$, $\Delta_k^\infty = \Phi_k$ is a direct consequence of Theorem~\ref{thm:k:stationary}.
Finally, since $W(z) \le 0$, for all $z \le 0$, we obtain
$\Delta^\rho_k \supseteq
	\left\{
	(\tau, \gamma, \sigma) \in \Phi_k \st
	2\tau \le \rho
	\right\}
	=
	\left[0, \frac{\rho}{2} \right] \times [0, \infty)^2$.

\subsection{Proof of Proposition~\ref{prop:k:matern}}
\label{prop:k:matern::proof}

Since $k$ is a stationary kernel structure with $\bar{k}_\eta(0_{2m+1}) = \tau + \sigma$ for every $\eta \in \Phi_k$, the claimed expression of $\Theta_k^\rho$ directly follows from Theorem~\ref{thm:k:stationary}.
In view of Theorem~\ref{thm:k:stationary}, $\Delta^0_k$ is the set of $\eta \in \Phi_k$ such that, for every $z \in \R^{2m+1}$,
$2\tau - 2 \tau \big( 1 + \gamma \sqrt{3} \norm*{z} \big) \exp \big(- \gamma \sqrt{3} \norm*{z} \big) \le \norm*{z}^2$
or, equivalently, for every $\zeta\in [0, \infty)$,
$g(\zeta) \coloneq 2\tau - 2 \tau \big( 1 + \gamma \sqrt{3} \zeta \big) \exp \big(- \gamma \sqrt{3} \zeta \big) - \zeta^2 \le 0$.
Notice that $g$ is smooth, $g(0)=0$, and its derivative is $g'(\zeta) = 2\zeta \big( 3 \gamma^2 \tau \exp(- \gamma \sqrt{3} \zeta) - 1 \big)$.
We have $g'(\zeta) \le 0$ for every $\zeta \ge 0$ if and only if $3 \gamma^2 \tau \le 1$ and, in this case, $g(\zeta)\le 0$ for all $\zeta \ge 0$.
Instead, if $3\gamma^2 \tau>1$, $g$ is increasing in $0$.
Hence, $\Delta_k^0 = \{ (\tau,\gamma,\sigma) \in \Phi_k \st 3 \gamma^2 \tau \le 1\}$.

\subsection{Proof of Proposition~\ref{prop:k:pillo}}
\label{prop:k:pillo::proof}

Since $k$ is a stationary kernel structure with $\bar{k}_\eta(0_{2m+1}) = \tau \sum_{t=0}^{m-p} e^{-\xi t}$ for every $\eta \in \Phi_k$, the claimed expression of $\Theta_k^\rho$ directly follows from Theorem~\ref{thm:k:stationary}.
Regarding $\Delta_k^\rho$, notice that we can write
$\bar{k}_\eta(z) = 2\gamma\tau \sum_{t=0}^{m-p} e^{-\xi t} w_{\tilde{\eta}}(z_t)$
where $w_{\tilde{\eta}}$ is the Gaussian Kernel (described in Proposition~\ref{prop:k:gaussian}) with parameter $\tilde{\eta} = ((2\gamma)^{-1}, \gamma, 0) \in \Phi_w$.
Since $2 (2\gamma)^{-1} \gamma=1$, then $\tilde{\eta} \in \Delta^0_w$.
Therefore, according to Theorem~\ref{thm:k:stationary}, we have $2 w_{\tilde{\eta}}(0_{2m+1}) - 2 w_{\tilde{\eta}}(z_t) \le \norm*{z_t}^2 \le \norm*{z}^2$.
Then, we obtain
\begin{align*}
	2 \bar{k}_\eta(0_{2m+1}) - 2 \bar{k}_\eta(z)
	&
	=
	2 \gamma\tau \sum_{t=0}^{m-p} e^{-\xi t} \big( 2 w_{\tilde{\eta}}(0_{2m+1}) - 2 w_{\tilde{\eta}}(z_t) \big)
	\\
	&
	\le
	\norm*{z}^2 \cdot 2 \gamma\tau \sum_{t=0}^{m-p} e^{-\xi t},
\end{align*}
from which the claim follows by Theorem~\ref{thm:k:stationary}, since $\sum_{t=0}^{m-p} e^{-\xi t}= \pi(\xi, p)$.

\subsection{Proof of Proposition~\ref{prop:k:sum}}
\label{prop:k:sum::proof}

First, let us focus on $\Theta^\rho_k$. Conditions~\eqref{eq:stability:kernel:cond} read
\begin{align}
	\forall a \in \R^{2m+1},
	& & |a|^2 \ge \nu
	& \Rightarrow \sum_{i=1}^q \tau_i w_{i,\eta_i} (a, a) \le |a|^2, \label{eq:proof:ksum:c1}
	\\
	\forall a \in \R^{2m+1},
	& & |a|^2 < \nu
	& \Rightarrow \sum_{i=1}^q \tau_i w_{i,\eta_i} (a, a) \le s. \label{eq:proof:ksum:c2}
\end{align}
Pick $\rho\in[0,\infty]$.
For every $i \in \{1, \ldots, q\}$, $\eta_i \in \Theta^\rho_{w_i}$ implies that there exist $\nu_i\in[0,\rho]\cap\R$ and $s_i \in [0, \infty)$ such that
\begin{align}
	\forall a \in \R^{2m+1},
	& & |a|^2 \ge \nu_i
	& \Rightarrow w_{i,\eta_i} (a, a) \le |a|^2, \label{eq:proof:ksum:c1i}
	\\
	\forall a \in \R^{2m+1},
	& & |a|^2 < \nu_i &\Rightarrow w_{i,\eta_i} (a, a) \le s_i. \label{eq:proof:ksum:c2i}
\end{align}
Let $\nu\coloneq\max_{i=1,\dots,q}\nu_i$ and $s\coloneq \max_{i=1,\dots,q}\max\{\nu_i,s_i\}$.
If $a\in\R^{2m+1}$ satisfies $|a|^2\ge \nu$, then $|a|^2\ge \nu_i$ for all $i=1,\dots,q$.
Hence, by~\eqref{eq:proof:ksum:c1i}, if $\sum_{i=1}^q\tau_i\le 1$, we obtain $\sum_{i=1}^q\tau_iw_{i,\eta_i}(a,a)\le |a|^2$.
Thus,~\eqref{eq:proof:ksum:c1} holds.
Consider now the case in which $|a|^2<\nu$.
For each $i$ such that $|a|^2<\nu_i$,~\eqref{eq:proof:ksum:c2i} yields $w_{i,\eta_i}(a,a)\le s_i\le s$; otherwise,~\eqref{eq:proof:ksum:c1i} yields $w_{i,\eta_i}(a,a)\le |a|^2 \le \nu_i\le \nu\le s$.
In both cases, $w_{i,\eta_i}(a,a)\le s$.
Then,~\eqref{eq:proof:ksum:c2} holds.

The claim regarding $\Delta^\rho_k$ follows by similar arguments.

\subsection{Proof of Proposition~\ref{prop:k:prod_stationary}}
\label{prop:k:prod_stationary::proof}

Let $\eta_1 \in \Theta^\rho_\ell$.
Then, there exist $s_1 \in (0, \infty)$ and $\nu_1 \in [0, \rho]$ such that
\begin{align*}
	\forall a \in \R^{2m+1},
	& & |a|^2 \ge \nu_1
	& \Rightarrow \ell_{\eta_1} (a,a) \le |a|^2,
	\\
	\forall a \in \R^{2m+1},
	& & |a|^2 < \nu_1
	& \Rightarrow \ell_{\eta_1} (a,a) \le s_1.
\end{align*}
Since $w_{\eta_2} (a,a) = \bar{w}_{\eta_2} (0_{2m+1})$, Conditions~\eqref{eq:stability:kernel:cond} read as
\begin{align*}
	\forall a \in \R^{2m+1},
	& & |a|^2 \ge \nu
	& \Rightarrow \ell_{\eta_1} (a,a) \bar{w}_{\eta_2} (0_{2m+1}) \le |a|^2,
	\\
	\forall a \in \R^{2m+1},
	& & |a|^2 < \nu
	& \Rightarrow \ell_{\eta_1} (a,a) \bar{w}_{\eta_2} (0_{2m+1}) \le s.
\end{align*}
Therefore, the second condition is satisfied with $s = s_1 \bar{w}_{\eta_2} (0_{2m+1})$ and $\nu = \nu_1$.
Instead, with $\nu = \nu_1$, the first condition is satisfied if $\bar{w}_{\eta_2}(0_{2m+1}) \le 1$.

\section*{References}
\bibliographystyle{IEEEtran}
\bibliography{ref}

\begin{IEEEbiography}[{\includegraphics[width=1in,height=1.25in,clip,keepaspectratio]{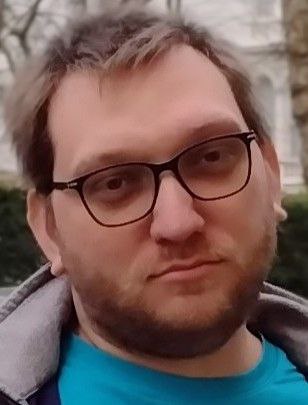}}]{Matteo Scandella} received the Ph.D. degree in engineering and applied science in 2020 from the University of Bergamo, Italy. Since February 2024, he has been with the Department of Management, Information and Production Engineering, University of Bergamo, Italy. From 2020 to 2024, he was with the Department of Electrical and Electronic Engineering, Imperial College London, UK. He serves as Associate Editor of European Journal of Control. His research interests include system identification, health monitoring and Bayesian methods.
\end{IEEEbiography}

\begin{IEEEbiography}[{\includegraphics[width=1in,height=1.25in,clip,keepaspectratio]{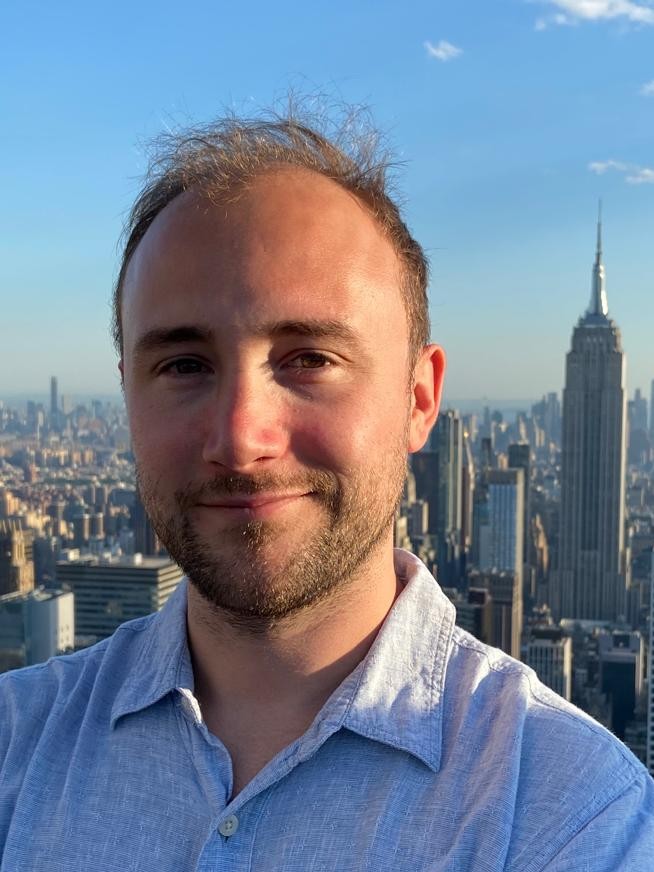}}]{Michelangelo Bin} (Member, IEEE) received the Ph.D. degree in Control Theory in 2019 from the University of Bologna, Italy. Since October 2022, he has been with the Department of Electrical, Electronic and Information Engineering, University of Bologna, Italy. From 2019 to 2023, he was with the Department of Electrical and Electronic Engineering, Imperial College London, UK. He serves as Associate Editor of Systems \& Control Letters. His research interests include systems theory, nonlinear control and regulation, and adaptive systems.
\end{IEEEbiography}

\begin{IEEEbiography}[{\includegraphics[width=1in,height=1.25in,clip,keepaspectratio]{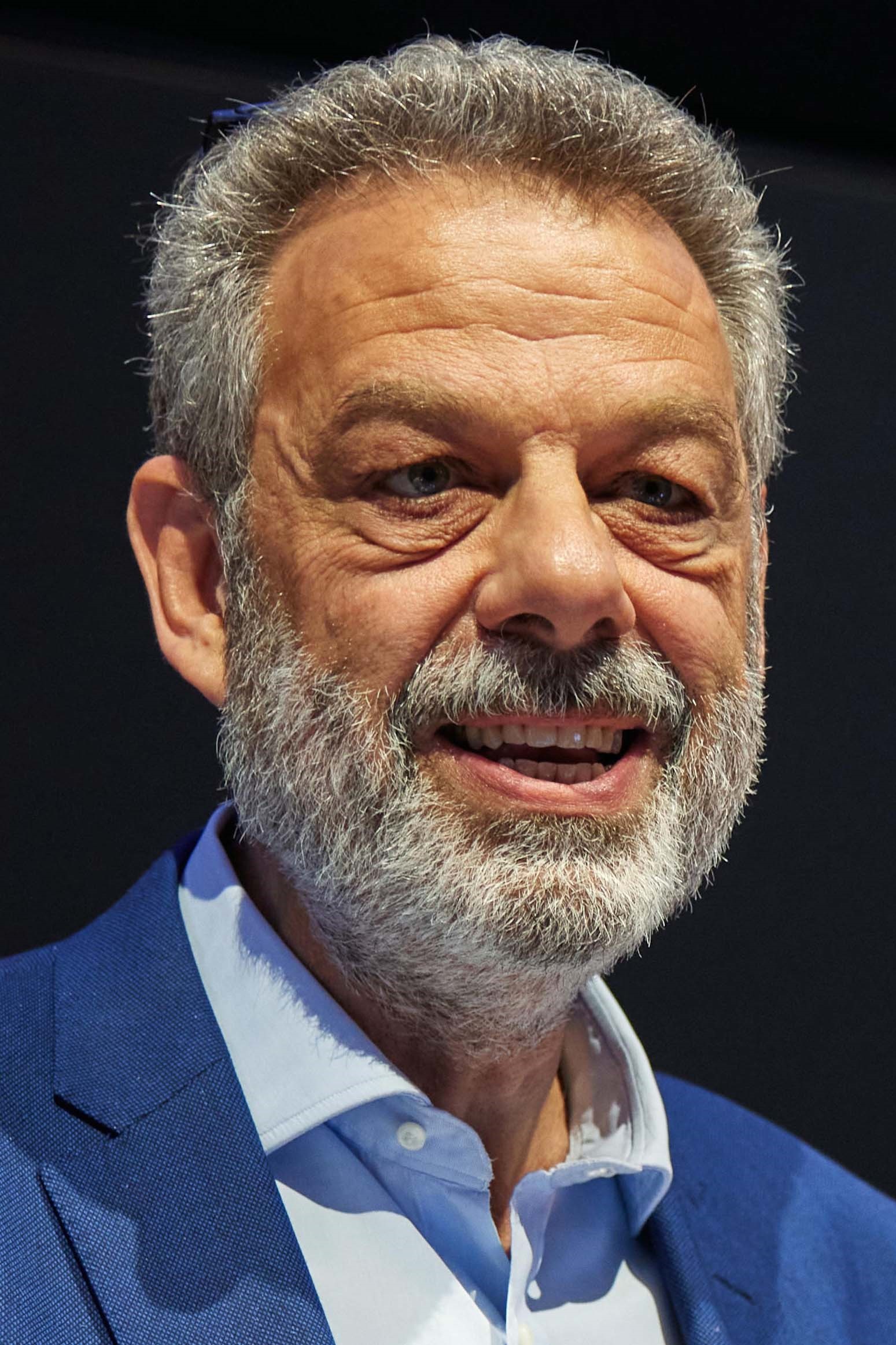}}]{Thomas Parisini} (Fellow, IEEE) received the Ph.D. degree in electronic engineering and computer science from the University of Genoa, Italy, in 1993. He was an Associate Professor with Politecnico di Milano, Milano, Italy. He currently holds the Chair of industrial control and is the Head of the Control and Power Research Group, Imperial College London, London, U.K. He also holds a Distinguished Professorship at Aalborg University, Denmark. Since 2001, he has been the Danieli Endowed Chair of automation engineering with the University of Trieste, Trieste, Italy, where from 2009 to 2012, he was the Deputy Rector. In 2023, he held a “Scholar-in-Residence”visiting position with Digital Futures-KTH, Stockholm, Sweden. He has authored or coauthored a research monograph in the Communication and Control Series, Springer Nature, and more than 400 research papers in archival journals, book chapters, and international conference proceedings. Dr. Parisini was the recipient of the Knighthood of the Order of Merit of the Italian Republic for scientific achievements abroad awarded by the Italian President of the Republic in 2023. In 2018 he received the Honorary Doctorate from the University of Aalborg, Denmark and in 2024, the IEEE CSS Transition to Practice Award. Moreover, he was awarded the 2007 IEEE Distinguished Member Award, and was co-recipient of the IFAC Best Application Paper Prize of the Journal of Process Control, Elsevier, for the three-year period 2011-2013 and of the 2004 Outstanding Paper Award of IEEE TRANSACTIONS ON NEURAL NETWORKS. In 2016, he was awarded as Principal Investigator with Imperial of the H2020 European Union flagship Teaming Project KIOS Research and Innovation Centre of Excellence led by the University of Cyprus with an overall budget of over 40 million Euros. He was the 2021-2022 President of the IEEE Control Systems Society and he was the Editor-in-Chief of IEEE TRANSACTIONS ON CONTROL SYSTEMS TECHNOLOGY (2009-2016). He was the Chair of the IEEE CSS Conference Editorial Board (2013-2019). Also, he was the associate editor of several journals including the IEEE TRANSACTIONS ON AUTOMATIC CONTROL and the IEEE TRANSACTIONS ON NEURAL NETWORKS. He is currently an Editor of Automatica and the Editor-in-Chief of the European Journal of Control. He was the Program Chair of the 2008 IEEE Conference on Decision and Control and General Co-Chair of the 2013 IEEE Conference on Decision and Control. He is a Fellow of IFAC. He is a Member of IEEE TAB Periodicals Review and Advisory Committee.
\end{IEEEbiography}

\end{document}